\documentclass{ifacconf}
\usepackage{graphicx}      % include this line if your document contains figures

\usepackage{amsmath,amssymb}
\usepackage{bm}
\usepackage{svg}
\usepackage{textcomp}
\usepackage{xcolor}
\usepackage{url}
\usepackage{algorithm}
\usepackage[noend]{algorithmic}% ht
\usepackage{xprintlen}
\usepackage{cases}
\usepackage{tabularx,ragged2e}
\usepackage{booktabs}
\usepackage{nccmath}
\usepackage{float}
\usepackage{mathtools}
\usepackage{verbatim}
\usepackage{import}

\renewcommand{\vec}[1]{\bm{#1}}

\def\myunderbar#1{\underline{\sbox\tw@{$#1$}\dp\tw@\z@\box\tw@}}

\newcommand{\vect}[1]{\boldsymbol{#1}}
\newcommand{\norm}[1]{ \lVert #1 \rVert}
\newcommand{\abs}[1]{ | #1 |}

\newcommand{\Nhop}[2]{\mathcal{N}_{#1}^{{#2}\text{-hop}}} 
\newcommand{\Neigh}[2]{\mathcal{N}_{#1}^{{#2}}} 
\newcommand{\estimate}[3]{\hat{{#1}}^{{#2}}_{{#3}}}
\newcommand{\state}[3]{{{#1}}^{{#2}}_{{#3}}}
\newcommand{\error}[3]{\Tilde{{#1}}^{{#2}}_{{#3}}}
\newcommand{\statediff}[5]{{#1}^{#2}_{#3}-{#4}^{#5}_{#3}}


\usepackage{natbib}        % required for bibliography
\begin{document}
\begin{frontmatter}

\title{Observer-based Control of Multi-agent Systems under STL Specifications\thanksref{footnoteinfo}} 

\thanks[footnoteinfo]{This work was supported in part by the Wallenberg AI, Autonomous Systems and Software Program (WASP) funded by the Knut and Alice Wallenberg (KAW) Foundation, Horizon Europe EIC project SymAware (101070802), the ERC LEAFHOUND Project, and the Swedish Research Council (VR).}

\author[First]{Tommaso Zaccherini}
\author[Second]{Siyuan Liu}
\author[First]{Dimos V. Dimarogonas} 

\address[First]{Division of Decision and Control Systems, KTH Royal Institute of Technology, Stockholm, Sweden. (e-mails: \{tommasoz, dimos\}@kth.se)}
\address[Second]{Control Systems Group, Department of Electrical Engineering, Eindhoven University of Technology, the Netherlands. (e-mail: s.liu5@tue.nl)}

% \author[First]{Tommaso Zaccherini} 
% \author[Second]{Siyuan Liu}
% \author[First]{Dimos V. Dimarogonas}
% \address[First]{Division of Decision and Control Systems, KTH Royal Institute of Technology, Stockholm, Sweden. E-mail: \{tommasoz, dimos\}@kth.se}
% \address[Second]{Department of Electrical Engineering, Control Systems Group, Eindhoven University of Technology, the Netherlands. Email: s.liu5@tue.nl}

\begin{abstract}                % Abstract of 50--100 words
This paper proposes a decentralized controller for \mbox{large-scale} heterogeneous \mbox{multi-agent} systems subject to bounded external disturbances, where agents must satisfy Signal Temporal Logic (STL) specifications requiring cooperation among non-communicating agents. To address the lack of direct communication, we employ a decentralized $k$-hop Prescribed Performance State Observer ($k$-hop PPSO) to provide each agent with state estimates of those agents it cannot communicate with. By leveraging the performance bounds on the state estimation errors guaranteed by the $k$-hop PPSO, we first modify the space robustness of the STL tasks to account for these errors, and then exploit the modified robustness to design a decentralized \mbox{continuous-time} feedback controller that ensures satisfaction of the STL tasks even under \mbox{worst-case} estimation errors. A simulation result is provided to validate the proposed framework.

\end{abstract}
\begin{keyword}
% Multi-agent systems, Distributed control and estimation, Control over networks, Control under communication constraints
Multi-agent systems, Signal Temporal Logic, Observer-based control, Task/communication graphs mismatch, Decentralized control and estimation.
% \red{took from Systems and signals, 2 missing but the others in this category seems unrelated or usually referred to some other topics like: Control under communication constraints, at the same time other categories do not have multi-agent systems}
\end{keyword}

\end{frontmatter}
%===============================================================================
\section{Introduction}
A \mbox{multi-agent} system (MAS) consists of a collection of interacting agents that collaborate or compete to accomplish specific objectives. Distributing intelligence and control among agents enhances the scalability, flexibility, and resilience of the overall system. Such decentralized coordination enables efficient problem-solving in complex and dynamic environments where single-agent systems (SAS) often face limitations. Owing to these advantages, considerable research efforts have focused on the control of MAS under various tasks, including formation control \citep{OH2015424}, consensus \citep{4118472}, and flocking \citep{olfati2006flocking}.

In recent years, the control of MAS has been further extended to handle more complex specifications expressed through temporal logics \citep{maler2004monitoring}. In particular, several studies have explored the use of Signal Temporal Logic (STL) to formally specify and guarantee \mbox{high-level} tasks in multi-agent coordination \citep{LINDEMANN2021100973, 10918825} and planning \citep{9345973, 9696363}.
However, these control methods heavily rely on the existence of communication links among the agents that must cooperate, limiting their applicability in scenarios where collaboration is required despite the absence of direct communication. To address this limitation, \cite{11222720} proposes a task decomposition algorithm that redefines tasks involving \mbox{non-communicating} agents to align them with the available communication topology. Nevertheless, this approach requires solving an optimization problem to decompose the task along a selected path, introducing an additional layer of complexity to the control problem and creating a dependency on the chosen path.

For this reason, in this paper we extend our previous works on \mbox{$k$-hop} \mbox{observer-based} control for MAS \citep{11018605,zaccherini2025robustestimationcontrolheterogeneous} to handle STL specifications. By leveraging the decentralized $k$-hop Prescribed Performance State Observer ($k$-hop PPSO) developed in \cite{zaccherini2025robustestimationcontrolheterogeneous}, each agent can estimate the states of the \mbox{non-communicating} agents with which it must collaborate, while ensuring that the predefined performance bounds specified at the design stage are satisfied. Based on the performance guarantees provided by the $k$-hop PPSO on the state estimation errors, the space robustness \citep{donze2010robust} of the STL tasks involving \mbox{non-communicating} agents is adjusted to account for the \mbox{worst-case} estimation errors. Then, a decentralized Prescribed Performance Controller (PPC) \citep{4639441} based on the modified robustness is proposed to ensure satisfaction of the task even when the \mbox{non-local} information is replaced by estimates.

The remainder of the paper is organized as follows. Section~\ref{Preliminaries} presents the preliminaries and the problem formulation. Section~\ref{Robustness adaptation} introduces the adjustment for the space robustness of the STL tasks to account for the state estimation errors. Section~\ref{Section: Decentralized controller} presents the proposed \mbox{observer-based} decentralized controller. Section~\ref{Section: Simulation} shows a case study. Section~\ref{Section: Conclusion} concludes the paper with future work directions.

\par

\section{Preliminaries and Problem Setting}\label{Preliminaries}
\textbf{Notation:} We denote by $\mathbb{R}$, $\mathbb R_{\ge 0}$, and $\mathbb R_{> 0}$ the sets of real, \mbox{non-negative} and positive real numbers, respectively. Let $|S|$, $S^c$, and $\partial S$ be the cardinality, complement and boundary of a set $S$. Given $N$ sets $\{S_1,\dots, S_N\}$, we denote their Cartesian product, intersection, and union by $\bigtimes^N _{i=1} S_i$, $\bigcap^N _{i=1} S_i$, and $\bigcup^N _{i=1} S_i$, respectively. For a set $S = \{s_1, \dots, s_n\}$, $\max_{i\in\{1,\dots,n\}} \{s_i\}$ denotes its maximum element.
We use $\top$ and $\bot$ to denote the logical values true and false.
Given a symmetric matrix $B \in \mathbb{R}^{n\times n}$, $\lambda_{\min}(B)$ denotes its minimum eigenvalue, and $B \succ 0$ indicates that $B$ is positive definite. For $x \in \mathbb{R}^n$, $\norm{x}= \sqrt{x^\top x}$. We use $f \in \mathcal{C}_1$ to indicate that a function $f$ is continuously differentiable in its domain. For $a \in \mathbb{R}$, $\exp{(a)}:= e^a$.
The diagonal matrix with diagonal entries $a_1,\dots , a_n$ is denoted by $\text{diag}(a_1, \dots, a_n)$. 
% Given a set $S = \{s_1, \dots, s_n\}$, we denote by $\max_{i} \{s_i\}$ the maximum element in $S$.
% We define with $\mathcal{K}$ and $\mathcal{KL}$ the following: $\mathcal{K} \!=\! \{\gamma : \mathbb R_{\ge 0}\rightarrow\mathbb R_{\ge 0}   :  \gamma \text{ is continuous, strictly increasing and } \gamma(0)=0\}$;  $\mathcal{KL} \!=\! \{\beta : \mathbb R_{\ge 0} \!\times \mathbb R_{\ge 0} \rightarrow\mathbb R_{\ge 0}  :$ for each fixed $s$, the map  $\beta(r,s)$  belongs to class  $\mathcal{K}$  with  respect to  $r$  and, for each fixed  nonzero $r$,  the map $\beta(r,s)$ is decreasing with respect to  $s$  and $\beta(r,s) \rightarrow 0 \text{ as } s \rightarrow \infty \}$.

\subsection{Signal Temporal Logic (STL)}
\label{Section Signal Temporal Logic}
Signal Temporal Logic (STL) is a formalism for specifying temporal and quantitative properties of \mbox{continuous-time} signals \citep{maler2004monitoring}. It combines logical predicates with temporal operators to encode constraints on both signal values and time intervals over which they must hold. Each predicate $\mu$ is defined by a continuously differentiable function $\mathcal{P} : \mathbb{R}^n \rightarrow \mathbb{R}$, such that $\mu = \top$ if $\mathcal{P}(x) \ge 0$ and $\mu = \bot$ otherwise, for $x \in \mathbb{R}^n$.
In this paper we consider \mbox{non-temporal} \eqref{eq: definition of non temporal formula} and temporal formulas \eqref{eq: definition of temporal formula}, defined recursively as:
\begin{equation}
    \label{eq: definition of non temporal formula}
    \psi ::= \top \ | \  \mu \  | \ \neg \mu \ | \ \psi_1 \land \psi_2,
\end{equation}
\begin{equation}
    \label{eq: definition of temporal formula}
    \phi ::= G_{[a,b]} \psi \ | \ F_{[a,b]} \psi  \ | \  F_{[\underline{a},\underline{b}]} G_{[\bar{a},\bar{b}]}  \psi, 
\end{equation}
where $\neg$, $\land$, $ G_{[a,b]}$, and $F_{[a,b]}$ denote the negation, conjunction, always and eventually operators, with $a, b \in \mathbb{R}_{\geq 0}$ and $ a \leq b$. $\psi_1$ and $\psi_2$ in \eqref{eq: definition of non temporal formula} are \mbox{non-temporal} formulas of class \eqref{eq: definition of non temporal formula}.
To quantify how robustly a trajectory $x: \mathbb{R}_{\geq 0} \rightarrow X \subseteq \mathbb{R}^n$ satisfies or violates an STL formula $\phi$, we use the concept of robust semantics \cite[Def.~3]{donze2010robust}, defined as:
% , the robust semantics of STL operators are defined as 
\begin{equation*}
    \begin{split}
        \rho^{\mu}(x,t)&:= \mathcal{P}(x(t)), \\
        \rho^{\neg \phi}(x,t)&:= -\rho^{\phi}(x,t),    \\
         \rho^{\phi_1 \wedge \phi_2}(x,t)&:= \min (\rho^{\phi_1}(x,t), \rho^{\phi_2}(x,t)), \\
        \rho^{F_{[a, b]} \phi}(x,t)&:=\max _{t_1 \in [t+a, t+b]} \rho^\phi(x, t_1),\\
        \rho^{G_{[a, b]} \phi}(x,t)&:= \min _{t_1 \in [t+a, t+b]} \rho^\phi(x, t_1),\\
        \rho^{F_{[\underline{a},\underline{b}]} G_{[\bar{a},\bar{b}]}  \phi} &:=  \max _{t_1 \in [t+\underline{a}, t+\underline{b}]} \min _{t_2 \in [t_1+\bar{a}, t_1+\bar{b}]} \rho^\phi(x, t_2),
    \end{split}
\end{equation*}
where $\rho^\phi(x, t): X \times \mathbb{R}_{\geq 0}  \rightarrow \mathbb{R}$ is the \textit{robustness function} of $\phi$, such that $\rho^\phi(x, t) > 0$ implies that $\phi$ is satisfied at time $t$. In order to develop a \mbox{continuous-time} feedback controller that requires the derivative of the robustness function (cf. Theorem \ref{Theorem on task satisfacion}), $\rho^{\phi_1 \wedge \phi_2}(x,t)$ is approximated by the smooth function $\bar{\rho}^{\phi_1 \wedge \phi_2}(x,t) = -\frac{1}{\eta} \ln(\exp(- \eta \rho^{\phi_1}(x,t)) + \exp(- \eta \rho^{\phi_2}(x,t)))$ as in \cite{7799279}. Note that, since $\bar{\rho}^{\phi_1 \wedge \phi_2}(x,t) \leq \rho^{\phi_1 \wedge \phi_2}(x,t)$ for all $\eta >0$, with equality when $\eta \rightarrow \infty$, $\bar{\rho}^{\phi_1 \wedge \phi_2}(x,t) > 0$ implies $\rho^{\phi_1 \wedge \phi_2}(x,t)> 0$.

As stated in \cite{LINDEMANN2021100973}, satisfaction of an STL formula $\phi$ with robustness $\rho^{\phi}(x,t)$ can be ensured  by enforcing the \mbox{non-temporal} robustness $\rho^\psi(x)$ to satisfy 
\begin{equation}
    \label{prescribed perfomance on the tasks. Preliminary section}
    -\gamma(t) + \rho^{\max} < \rho^\psi(x) < \rho^{\max},
\end{equation}
where $\rho^{\max} \in \mathbb{R}_{\geq 0}$, $\gamma(t)= (\gamma^0 - \gamma^{\infty}) \exp(-l t) +\gamma^{\infty}$ with $t\in \mathbb{R}_{\geq 0}$, and the parameters $l\in \mathbb{R}_{\geq 0}$, $\gamma^0$ and $ \gamma^{\infty} \in \mathbb{R}_{> 0}$, can be chosen such that $\rho^\phi(x, t) > 0$ is guaranteed by imposing \eqref{prescribed perfomance on the tasks. Preliminary section}.

\subsection{Multi-agent systems}
\label{Section: multi-agent system}
\vspace{-0.3cm}
Consider a heterogeneous MAS of $N$ agents $\mathcal{V} = \{1,\dots, N\}$, where the dynamics of each agent $i \in \mathcal{V}$ is given by:
\begin{equation}
    \label{eq: agent's dynamic}
    \dot x_i(t) =  f_i(x_i(t)) + g_i(x_i(t))u_i(t) + w_i(\vect{x},t),
\end{equation}
where $x_i \in \mathbb{R}^{n_i}$ and $u_i \in \mathbb{R}^{m_i}$ denote the state and input of agent $i$, $f_i:\mathbb{R}^{n_i} \rightarrow \mathbb{R}^{n_i}$ is the flow drift, $g_i:\mathbb{R}^{n_i} \rightarrow \mathbb{R}^{n_i \times m_i}$ is the input matrix and $w_i: \bigtimes^N _{i=1}\mathbb{R}^{n_i} \times \mathbb{R}_{\geq 0} \rightarrow \mathbb{R}^{n_i}$ represents external disturbances acting on $i$. The global state and input vectors of the MAS are defined as $\vect{x}= \left[x_1^\top, \dots, x_N^\top \right]^\top \in \mathbb{R}^n$ and $ \vect{u}=\left[u_1^\top, \dots, u_N^\top\right]^\top \in \mathbb{R}^m$, where $n = \sum_{i=1}^{N} n_i$ and $m = \sum_{i=1}^{N} m_i$.
\begin{assumption}
    \label{Assumption on existence of a solution}
    (i) $f_i:\mathbb{R}^{n_i} \rightarrow \mathbb{R}^{n_i}$ and $g_i:\mathbb{R}^{n_i} \rightarrow \mathbb{R}^{n_i \times m_i}$ are locally Lipschitz continuous; (ii) $g_i(x_i)g_i(x_i)^\top $ is positive definite for all $x_i \in \mathbb{R}^n_i$ (iii) $w_i: \bigtimes^N _{i=1}\mathbb{R}^{n_i} \times \mathbb{R}_{\geq 0} \rightarrow \mathbb{R}^{n_i}$ is continuous and uniformly bounded. 
    % in $\bigtimes^N _{i=1}\mathbb{R}^{n_i} \times \mathbb{R}_{\geq 0}$.
\end{assumption}

\textbf{Communication graph:} the \mbox{information-exchange} structure among the agents is modeled by an undirected graph $\mathcal{G}_C = (\mathcal{V}, \mathcal{E}_C)$, where $\mathcal{V}$ denotes the set of agents and $\mathcal{E}_C \subseteq \mathcal{V} \times \mathcal{V}$ represents the set of communication links.  
A path between two agents $i,j \in \mathcal{V}$ is defined as a sequence of non-repeating edges that allows $j$ to be reached from $i$. A \mbox{$k$-hop} path between $i$ and $j$ corresponds to a path of length $k$, i.e., a sequence of exactly $k$ edges connecting agent $i$ to $j$. Given $i,j\in \mathcal{V}$,  we denote by $D_{ij}$ the length of the shortest path in $\mathcal{G}_C$ between $i$ and $j$. 
For each agent $i \in \mathcal{V}$, $\Nhop{i}{k} := \{ j \in \mathcal{V} \,|\, \exists \, p\text{-hop path from $j$ to $i$ with } 2 \leq p \leq k \}$ defines the set of $k$-hop neighbors of $i$. The elements of this set are labeled as $\Nhop{i}{k} = \{N^i_1, \dots, N^i_{\eta_i}\}$, where each $N^i_j \in \mathcal{V}$ denotes the global index of the \mbox{$j$-th} \mbox{$k$-hop} neighbor of $i$, and $\eta_i = |\Nhop{i}{k}|$ indicates the total number of such neighbors.   
The set of direct ($1$-hop) neighbors of agent $i$ in the communication graph is defined as $\Neigh{i}{C} := \{ j \in \mathcal{V} \,|\, (i,j) \in \mathcal{E}_C \lor j = i\}$ .

\textbf{Task dependency graph:} the MAS in \eqref{eq: agent's dynamic} is subject to a global STL specification $\phi =\land_{i=1}^N \phi_i$, where each $\phi_i$ denotes the local STL specification assigned to agent $i \in \mathcal{V}$. 
To capture the dependency of task $\phi_i$ from agents $j \in \mathcal{V} \setminus \{i\}$, a \emph{task dependency graph} $\mathcal{G}_T = (\mathcal{V}, \mathcal{E}_T)$ is introduced, where an edge $(i,j) \in \mathcal{E}_T$ if and only if $\phi_i$ depends on $x_j$. We denote by $\mathcal{N}^T_i := \{ j \in \mathcal{V} \,|\, (i,j) \in \mathcal{E}_T \lor j = i\}$ the set of \mbox{out-neighbors} of agent $i$ in the task graph $\mathcal{G}_T$. 

\begin{assumption}
    \label{Assumption on neighbors}
    $\mathcal{G}_C$ is a \mbox{time-invariant}, connected and undirected graph, and each $i\in \mathcal{V}$ knows $\Neigh{i}{C}$ and $\Nhop{i}{k}$.
\end{assumption}
\begin{assumption}
    \label{Assumption on Task graph}
    $\mathcal{G}_T$ is a directed acyclic graph. In this context, \mbox{self-loops} do not constitute a cycle and thus are allowed in $\mathcal{G}_T$.
\end{assumption}
Assumption~\ref{Assumption on neighbors} is not restrictive, since distributed neighborhood discovery algorithms have been extensively studied in the sensor network literature \citep{994183}. Assumption~\ref{Assumption on Task graph} is not limiting, as task assignment can be defined at design stage.

\subsection{Cluster-induced graph}
\label{Section: cluster induced graph}
\vspace{-0.1cm}
Before introducing the concept of clusters, let's define the intersection among graphs as follows.
\begin{definition}
    \label{def: graphs intersection}
    Given $M$ graphs $\mathcal{G}_i = (\mathcal{V}_i, \mathcal{E}_i)$, $i \in \{1,\dots, M\}$, the intersection graph is defined as $\mathcal{G}^{\cap} = (\mathcal{V}^{\cap}, \mathcal{E}^{\cap})$, where $\mathcal{V}^{\cap} = \bigcap_{i =1}^M \mathcal{V}_i$ and $ \mathcal{E}^{\cap} = \bigcap_{i =1}^M \mathcal{E}_i$. For the intersection among directed and undirected graphs, each edge $(j,l)$ of the undirected graph must be considered as existing in both directions, i.e., both as $(j,l)$ and $(l,j)$.
\end{definition}

The intersection between the task and communication graphs, i.e., the graph $\mathcal{G}_T \cap \mathcal{G}_C$ defined according to Definition~\ref{def: graphs intersection}, partitions the MAS into $N'$ connected components $\mathcal{C}_l$, referred to as \textbf{clusters}. Formally, for $l\in \mathcal{L}:=\{1,\dots, N'\}$, $\mathcal{C}_l = (\mathcal{V}_l, \mathcal{E}_l)$, where $\mathcal{V}_l  = \{l_1, \dots, l_{v_l}\} \subseteq \mathcal{V}$ denotes the set of $v_l$ agents in the $l$-th cluster, and $\mathcal{E}_l\subseteq \mathcal{E}_C \cap \mathcal{E}_T$ denotes the set of edges among them.  
From the definition of \textit{connected components}, $\cup_{l = 1}^{N'} \mathcal{V}_l= \mathcal{V}$ and $\mathcal{V}_i \cap \mathcal{V}_j = \emptyset$ hold for all $i,j \in \mathcal{L}$ with $i \neq j$. As a result, each $\mathcal{C}_l$ is associated with $\phi^{c}_l = \land_{l_i \in \mathcal{V}_l} \phi_{l_i}$, and the global STL task $\phi$ can be rewritten as $\phi = \land_{l=1}^{N'} {\phi}^c_l$. 

The \mbox{cluster-induced} graph is then defined as follows.
\begin{definition}
    \label{def: cluster induced graph}
    Given a MAS \eqref{eq: agent's dynamic}, with communication graph $\mathcal{G}_C$ and task graph $\mathcal{G}_T$, the \textit{\mbox{cluster-induced} graph} is defined as $\mathcal{G}' = (\mathcal{C}', \mathcal{E}')$, where each node in $\mathcal{C}'$ is associated to a cluster $\mathcal{C}_l$, with $l \in \mathcal{L}$, and an edge $(\mathcal{C}_l, \mathcal{C}_j) \in \mathcal{E}'$ exists if and only if there exists a task $\phi_{l_i}$, with ${l_i} \in \mathcal{V}_l$, whose satisfaction depends on the state of an agent $j_q \in \mathcal{V}_j$, with $(l_i,j_q) \in \mathcal{E}_T$.
\end{definition}
\begin{remark}
    \label{remark on the acyclicity of the cluster induced graph}
    Note that under Assumption~\ref{Assumption on Task graph}, $\mathcal{G}'$ is a directed acyclic graph. Indeed, since $\mathcal{G}'$ is obtained by grouping the agents according to the clusters $\mathcal{C}_l$, no new directed cycles can be introduced.
\end{remark}

\begin{assumption}
    \label{Assumption on the non-existance of a task not in communication inside a cluster}
     % Agents within the same cluster that are task neighbors are assumed to be in communication. Formally, $\Neigh{l_i}{T} \cap \mathcal{V}_l \subseteq \Neigh{l_i}{C} \cap \mathcal{V}_l$ holds for all $l_i\in \mathcal{V}_l$, with $l \in \mathcal{L}$.
     Task neighbors within the same cluster are assumed to be in communication. Formally, $\Neigh{l_i}{T} \cap \mathcal{V}_l \subseteq \Neigh{l_i}{C} \cap \mathcal{V}_l$ for all $l_i\in \mathcal{V}_l$, with $l \in \mathcal{L}$.
\end{assumption}
Assumption \ref{Assumption on the non-existance of a task not in communication inside a cluster} ensures that all {intra-cluster} task dependencies correspond to communication links, allowing the observer and control designs to be decoupled. 
Since $\Neigh{l_i}{T} \cap \mathcal{V}_l \subseteq \Neigh{l_i}{T}$ and $\Neigh{l_i}{C} \cap \mathcal{V}_l \subseteq \Neigh{l_i}{C}$, Assumption~\ref{Assumption on the non-existance of a task not in communication inside a cluster} does not enforce $\Neigh{l_i}{T} \setminus \Neigh{l_i}{C} = \emptyset$, nor $\Neigh{l_i}{C} \setminus \Neigh{l_i}{T} = \emptyset$. Consequently, agent $l_i\in \mathcal{V}_l$ may still have tasks involving agents $j_q \in \mathcal{V}_j$, $j \in \mathcal{L}\setminus \{l\}$, with $j_q \in \Neigh{{l_i}}{T}$ and $j_q \notin \Neigh{{l_i}}{C}$. 

\vspace{-0.1cm}
\subsection{Individual/collaborative tasks}
\label{Section: collaborative and non-collaborative tasks}
\vspace{-0.2cm}
Consider the local task $\phi_i$ assigned to agent $i \in \mathcal{V}$. As introduced in Section~\ref{Section Signal Temporal Logic}, its satisfaction is evaluated through the robustness function $\rho^{\phi_i}(\vect{x}_{\phi_i}, t)$, where $\vect{x}_{\phi_i}$ collects the state of agents involved in $\phi_i$, i.e., $x_j$ with $j\in \Neigh{i}{T}$.
Therefore, satisfaction of $\phi_i$ may depend on agent $i$ and possibly other agents in $\mathcal{V}\setminus \{i\}$. We say that $\phi_i$ is \textbf{individual} if $\Neigh{i}{T} = \{i\}$, and \textbf{collaborative} otherwise.
Collaborative tasks $\phi_i$ may involve both agents that communicate with $i$ and agents that do not, i.e., $j \in \Neigh{i}{T} \cap \Neigh{i}{C}$ and $j \in \Neigh{i}{T} \setminus \Neigh{i}{C}$. Thus, since $\vect{x}_{\phi_i}$ may be partially unavailable to $i$, we define its local estimate as $\hat{\vect{x}}_{\phi_i}$, and the corresponding robustness estimate as $\rho^{\psi_i}(\hat{\vect{x}}_{\phi_i}, t)$.
Let $\Bar{\vect{x}}_{\phi_i}$ denote the components of $\hat{\vect{x}}_{\phi_i}$ corresponding to known states, i.e., $x_j$ with $j \in \Neigh{i}{C} \cap \Neigh{i}{T}$, and let $\estimate{\vect{x}}{i}{\phi_i}$ denote those corresponding to estimated states, i.e, $x_j$ with $j \in \Neigh{i}{T} \setminus \Neigh{i}{C}$.

\begin{remark}
    \label{Remark on robustness equivalence under full state knowledge}
    $\rho^{\psi_i}(\vect{x}_{\phi_i},t) = \rho^{\psi_i}(\hat{\vect{x}}_{\phi_i},t)$ under full state knowledge.
\end{remark}
Let $N^{T}_{i} := \sum_{j \in \Neigh{i}{T}} n_j$, $N^{TC}_{i} := \sum_{j \in \Neigh{i}{C}\cap \Neigh{i}{T}} n_j$, and $N^{T\setminus C}_{i} := \sum_{j \in \Neigh{i}{T} \setminus \Neigh{i}{C}} n_j$. Then, $\vect{x}_{\phi_i}, \hat{\vect{x}}_{\phi_i}\in \mathbb{R}^{N^{T}_{i}}$, $\Bar{\vect{x}}_{\phi_i} \in \mathbb{R}^{N^{TC}_{i}}$ and $\estimate{\vect{x}}{i}{\phi_i} \in \mathbb{R}^{N^{T\setminus C}_{i}}$. In the sequel, let us denote by $\vect{x}_{\phi^c_l}= [x^{\top}_{l_1}, \dots x^{\top}_{l_{v_l}}]^{\top}$ the stacked vector of the state of agents in cluster $\mathcal{C}_l$, and by $\Hat{\vect{x}}_{\phi^c_l} = \Bigl[\estimate{\vect{x}}{{l_1} \ \top}{{\phi_{l_1}}}, \dots, \estimate{\vect{x}}{{l_{v_l}} \ \top}{{\phi_{l_{v_l}}}} \Bigr]^\top $ the one of the state estimates performed by agent $l_i \in \mathcal{V}_l$ regarding agents $j \in \Neigh{l_i}{T} \setminus \Neigh{l_i}{C}$, i.e., of the agents involved in $\phi_{l_i}$ but not in communication with $l_i$ itself.

To simplify the notation, we assume without loss of generality that $n_i=1$ for all $i\in \mathcal{V}$ in the following sections. Nonetheless, the results can be extended to higher dimensions by appropriate use of the Kronecker product.

\subsection{State observer}
\label{Section: Distributed observer}
To enable each agent $i \in \mathcal{V}$ to estimate the state of agents $j \in \Neigh{i}{T} \setminus \Neigh{i}{C}$, we employ the decentralized \mbox{$k$-hop} Prescribed Performance State Observer (\mbox{$k$-hop} PPSO) proposed in \cite{zaccherini2025robustestimationcontrolheterogeneous}, with 
\begin{equation}
    \label{selection of the parameter k}
    k = \max_{i \in \mathcal{V}, j \in (\mathcal{N}^T_i \setminus \mathcal{N}^C_i)} D_{ij},
\end{equation}
where $D_{ij}$ is defined as in Section~\ref{Section: multi-agent system}. 
For clarity, the observer structure is recalled in this subsection.

Denote by $\state{\vect{x}}{i}{} = [\state{x}{\top}{N_1^i}, \dots ,\state{x}{\top}{N^i_{\eta_i}}]^\top$ the vector of states $x_{N^i_{j}}$ of the \ensuremath{k}\text{-hop} neighbors of agent $i$, and let $\estimate{\vect{x}}{i}{} = [\estimate{x}{i \ \top}{N_1^i}, \dots ,\estimate{x}{i \ \top}{N^i_{\eta_i}}]^\top$ be the vector containing their estimate carried out by $i$. The state estimation error is then defined as $\error{\vect{x}}{i}{} =\estimate{\vect{x}}{i}{} -  \state{\vect{x}}{i}{} =[\error{x}{i \ \top}{N_1^i}, \dots ,\error{x}{i \ \top}{N^i_{\eta_i}}]^\top$, where $\error{x}{i}{N_j^i} = \estimate{x}{i}{N_j^i} -\state{x}{}{N_j^i}$ for all $N_j^i \in \Nhop{i}{k}$. 
% For all $i \in  \mathcal{V}$ and $r \in \Nhop{i}{k}$, let $\delta_{r}^i(t)$ and $\rho_{r}^i(t)$ be prescribed performance functions satisfying $\lVert\vect{\rho}_{r}(t)\rVert  \leq \lambda_{\text{min}}(\kc{M}{r}) \min_{i \in \Nhop{r}{k}}\{\delta_{r}^{i}(t)\}$, with $\vect{\rho}_{r}(t)= \Bigl[\rho_{r}^{N^r_1}(t),\dots , \rho_{r}^{N^r_{\eta_r}}(t)\Bigr]^\top$ \red{and $\kc{M}{r} = \kc{L}{r} + \kc{H}{r}$, where $\kc{L}{r}$ is the Laplacian matrix of the sub-graph $\mathcal{G}_C^r = ( \Nhop{r}{k}, \mathcal{E}_C^r)$ induced by the \mbox{$k$-hop} neighbors of agent $r$, with $\mathcal{E}_C^r = \{(p,q) \in \mathcal{E}_C: \{p,q\} \subseteq \Nhop{r}{k}\}$, and $\kc{H}{r}:= \text{diag}(|\Neigh{N_1^r}{C}\cap \Neigh{r}{C}|, \dots, |\Neigh{N_{\eta_r}^r}{C}\cap \Neigh{r}{C}|) \in \mathbb{R}^{\eta_r \times \eta_r}$. According to \cite{11018605}, $\lambda_{\text{min}}(\kc{M}{r})> 0$.} 
% Formally, the prescribed performance functions $\delta_{r}^i(t)$ and $\rho_{r}^i(t)$ are defined as follows.
For all $i \in  \mathcal{V}$ and $r \in \Nhop{i}{k}$, let $\delta_{r}^i(t)$ and $\rho_{r}^i(t)$ be prescribed performance functions satisfying $\lVert\vect{\rho}_{r}(t)\rVert  \leq \alpha \ {\min}_{i \in \Nhop{r}{k}}\{\delta_{r}^{i}(t)\}$, where $\vect{\rho}_{r}(t)= \Bigl[\rho_{r}^{N^r_1}(t),\dots , \rho_{r}^{N^r_{\eta_r}}(t)\Bigr]^\top$ and $\alpha \in \mathbb{R}_{> 0}$ is a constant parameter related to the communication graph topology \citep{11018605}. 
Formally, the prescribed performance functions $\delta_{r}^i(t)$ and $\rho_{r}^i(t)$ are defined as follows.
\begin{definition}[\cite{zaccherini2025robustestimationcontrolheterogeneous}]
    \label{definition of prescribed performance function}
    A function $\rho :\mathbb{R}_{\geq 0} \rightarrow \mathbb{R}$ is a prescribed performance function if it satisfies, for all $t \in \mathbb{R}_{\geq 0}$: (i) $\rho(t) \in \mathcal{C}^1$; (ii)  $0 <\rho(t) \leq \overline{\rho}$ for some $ \overline{\rho} <  \infty $ and (iii) $|\dot{\rho}(t)|\leq \dot{\overline{\rho}}$ for some $\dot{\overline{\rho}}<\infty$. 
\end{definition}

For all $i \in \mathcal{V}$ and $ r \in \Nhop{i}{k}$, let $\estimate{x}{i}{r}$ evolve according to
\begin{equation}
    \label{observer dynamics}
    \dot{\hat{x}}_{r}^{i} = - \rho_{r}^i(t)^{-1} J(e_{r}^i) \epsilon_{r}^{i}(t),
\end{equation}
where $J(e_{r}^i) = \ln(\frac{2}{1-(e_{r}^i)^2})$, $e_{r}^i = \rho^{i}_r(t)^{-1}\xi_{r}^{i}$ with $\xi^{i}_{r} = \sum_{l\in(\Neigh{i}{C} \cap \Nhop{r}{k})} (\statediff{\hat{x}}{i}{r}{\hat{x}}{l}) + |\Neigh{i}{C} \cap \Neigh{r}{C}|  (\hat{x}^{i}_r - x_r)$, and $\epsilon_{r}^{i}(t) = \ln(\frac{1+e_{r}^i}{1-e_{r}^i})$. Then, according to \cite[Thm.~2, Rem.~7]{zaccherini2025robustestimationcontrolheterogeneous}, the observer in \eqref{observer dynamics} guarantees $|\error{x}{i}{r}(t)| <\delta^{i}_r(t)$ for all $i \in \mathcal{V}$, $r \in \Nhop{i}{k}$ and $t\in\mathbb{R}_{\geq 0}$.

\begin{remark}
    As reported in \cite[Thm.~2, Rem.~7]{zaccherini2025robustestimationcontrolheterogeneous}, the result above holds as long as {$g_i(x_r)u_i$} is bounded and the dynamics \eqref{eq: agent's dynamic} remain within a bounded set $\mathcal{X}_i \subset \mathbb{R}^{n_i}$. As shown in the proof of Theorem~\ref{Theorem on task satisfacion}, this condition is not restrictive in the proposed framework.
\end{remark}
Once every $i \in \mathcal{V}$ can estimate the state of agents $j \in \mathcal{N}^T_i \setminus \mathcal{N}^C_i$, the control problem can be formulated as follows.
\begin{problem}
    \label{Problem: first problem formulation}
    Consider a heterogeneous MAS \eqref{eq: agent's dynamic} with connected communication graph $\mathcal{G}_C$ and global STL specification $\phi = \land_{i=1}^N \phi_i$, where $\phi_i$ denotes the local STL task of agent $i$. For each $i \in \mathcal{V}$, synthesize a decentralized controller $u_i$ that guarantees satisfaction of $\phi$ using only locally available information.
\end{problem}

\section{Observer-based Task Satisfaction}\label{Robustness adaptation}
As introduced in Section~\ref{Section Signal Temporal Logic}, to ensure satisfaction of the STL formula $\phi_i$, the \mbox{non-temporal} robustness $\rho^{\psi_i}(\vect{x}_{\phi_i})$ must satisfy:
\begin{equation}
    \label{eq: robustness inequality no estimations ivolved}
     -\gamma_{\psi_i}(t) + \rho_{\psi_i}^{\max} < \rho^{\psi_i}(\vect{x}_{\phi_i}) < \rho_{\psi_i}^{\max},
\end{equation}
where $\vect{x}_{\phi_i}$ collects the states involved in task $\phi_i$, and $\gamma_{\psi_i}(t)$ is a positive, decreasing exponential function as in \eqref{prescribed perfomance on the tasks. Preliminary section}, where $\gamma^0_{\psi_i}$, $\gamma^{\infty}_{\psi_i}$, $l_{\phi_i}$, $t_{\phi_i}^*$, $r_{\phi_i}$, and $\rho_{\psi_i}^{\max}$ are defined and tuned as in \cite{10918825}[Remark 23]. 
\begin{remark}
      $\gamma_{\psi_i}(t)$ in \eqref{prescribed perfomance on the tasks. Preliminary section} is a valid prescribed performance function according to Definition~\ref{definition of prescribed performance function}. 
\end{remark}

Since the robustness $\rho^{\psi_i}$ may depend on estimated rather than actual states,  inequality \eqref{eq: robustness inequality no estimations ivolved} cannot be enforced directly. However, since the \mbox{$k$-hop} PPSO introduced in Section~\ref{Section: Distributed observer} guarantees prescribed performances on the state estimation errors regardless of the system trajectory, we can formulate an analogous constraint on $\rho^{\psi_i}(\hat{\vect{x}}_{\phi_i})$ that ensures task satisfaction under the \mbox{worst-case} estimation error.
To this end, we introduce the following assumption.
\begin{assumption}
    \label{assumption on robustness lower bound}
    For all $i \in \mathcal{V}$, there exists a prescribed performance function $\rho_{\psi_i}^t:\mathbb{R}_{\geq 0} \rightarrow \mathbb{R}$, defined as in Definition \ref{definition of prescribed performance function}, such that
    \begin{equation}
        \label{eq: lower bound on robustness base on the estimation error}
        0 < {\rho}^{\psi_i}(\hat{\vect{x}}_{\phi_i}(t)) - \rho_{\psi_i}^t(t) \implies 0 < \rho^{\psi_i}(\vect{x}_{\phi_i}(t)).
    \end{equation} 
\end{assumption}
Assumption~\ref{assumption on robustness lower bound} is not restrictive in practice. Indeed, since: (i) $\vect{x}^i =  \estimate{\vect{x}}{i}{} - \error{\vect{x}}{i}{}$ holds by definition; (ii) $\norm{\error{\vect{x}}{i}{}} < \norm{[\delta^{i}_{N_1^i}(t), \dots, \delta^{i}_{N_{\eta_i}^i}(t)]}$ is guaranteed for all $i \in \mathcal{V}$ by the \mbox{$k$-hop} PPSO in \eqref{observer dynamics}; and (iii) $\norm{[\delta^{i}_{N_1^i}(t), \dots, \delta^{i}_{N_{\eta_i}^i}(t)]}$ is a prescribed performance function \citep[Lemma 2]{zaccherini2025robustestimationcontrolheterogeneous}; Assumption~\ref{assumption on robustness lower bound} holds for common STL tasks.
For clarity, an example of a formation/containment task is presented below.
\begin{example}
    \label{example: funnel modification}
    Consider a path graph with $N=4$ agents and denote the global state as $\vect{x}=[x_1 \ x_2\ x_3\ x_4]^\top \in \mathbb{R}^n$, where $n=4$. Assume $\Nhop{1}{3} =\{3,4\}$ and $\Neigh{1}{} =\{2\}$. Suppose agent $1$ is assigned a task $\phi_1$ with non-temporal robustness function $\rho^{\psi_1}([x_1, x_3]^\top) := r_{13}^2 -\abs{x_1 - x_3 - d_{13}}^2$, where $r_{13}, d_{13} \in \mathbb{R}_{\geq 0}$. Since $\state{x}{}{3} = \estimate{x}{1}{3} - \error{x}{1}{3}$ from definition,  
    \begin{equation}
        \label{eq in example: rewritten stl task}
         r_{13}^2 -\abs{x_1 - \estimate{x}{1}{3} + \error{x}{1}{3} - d_{13}}^2 > 0
    \end{equation}
    holds for all $t \in \mathbb{R}_{\geq 0}$. Then, since $\abs{x_1- \estimate{x}{1}{3} - d_{13} + \error{x}{1}{3}}^2 \leq (\abs{x_1 - \estimate{x}{1}{3}- d_{13}} + \abs{\error{x}{1}{3}})^2$ is valid from triangle inequality, \eqref{eq in example: rewritten stl task} is satisfied if $r_{13}^2 - (\abs{x_1 - \estimate{x}{1}{3}- d_{13}} + \abs{\error{x}{1}{3}})^2 > 0$. Since $r_{13}^2 - \abs{x_1 - \estimate{x}{1}{3}- d_{13}}^2 - \abs{\error{x}{1}{3}}^2 - 2\abs{x_1 - \estimate{x}{1}{3}- d_{13}}\abs{\error{x}{1}{3}} > 0$ holds in $\abs{x_1 - \estimate{x}{1}{3}- d_{13}} \in (0, - \abs{\error{x}{1}{3}} + r_{13})$, by introducing the bound resulting from the \mbox{$k$-hop} PPSO, i.e., $0 < \abs{\error{x}{1}{3}} < \delta_3^1(t)$,
    we derive that the previous inequality is satisfied if $\abs{x_1 - \estimate{x}{1}{3}- d_{13}} < - \delta_3^1(t) + r_{13}$, with $ \delta_3^1(t) < r_{13}$. By squaring both sides, we notice that $0 < \rho^{\psi_1}([x_1, \hat{x}^1_3]^\top) - \rho_{\psi_i}^t$ implies \eqref{eq in example: rewritten stl task} with $\rho^{\psi_1}([x_1, \hat{x}^1_3]^\top) = r_{13}^2 - \abs{x_1 - \estimate{x}{1}{3} -d_{13}}^2$ and $\rho_{\psi_i}^t(t) = 2 \delta_3^1(t)r_{13} -  \delta_3^1(t)^2$. Therefore, Assumption~\ref{assumption on robustness lower bound} holds.
    An analogous statement can be established for $n>1$ using the norm operator.
\end{example}
Given Assumption \ref{assumption on robustness lower bound}, guaranteeing the satisfaction of $\phi_i$ in presence of state estimates reduces to enforcing $-\gamma_{\psi_i}(t) + \rho^{\max}_{\psi_i} < \rho^{\psi_i}(\hat{\vect{x}}_{\phi_i}) - \rho_{\psi_i}^t(t) < \rho^{\max}_{\psi_i} -\rho_{\psi_i}^t(t)$, where $\gamma_{\psi_i}(t)$ is defined as in \eqref{eq: robustness inequality no estimations ivolved} and $\rho_{\psi_i}(\hat{\vect{x}}_{\phi_i})$ as in Section \ref{Section: collaborative and non-collaborative tasks}. Equivalently, the following condition must be enforced with $\Gamma_{\psi_i}(t) = \gamma_{\psi_i}(t) - \rho_{\psi_i}^t(t)$:
\begin{equation}
    \label{eq: inequality on robustness to guarantee task satisfaction, observer not in design stage of gamma}
    -\Gamma_{\psi_i}(t) < \rho^{\psi_i}(\hat{\vect{x}}_{\phi_i}) - \rho^{\max}_{\psi_i} < 0.
\end{equation}
Note that \eqref{eq: inequality on robustness to guarantee task satisfaction, observer not in design stage of gamma} resembles the structure in \eqref{eq: robustness inequality no estimations ivolved}.
Thus, to guarantee the possibility of applying a PPC inspired controller, the following assumption is introduced on $\Gamma_{\psi_i}(t)$.
\begin{assumption}
    \label{Assumption on the positive difference between gamma and observer uncertainties}
    $\Gamma_{\psi_i}(t) > 0$ for all $t\in\mathbb{R}_{\geq 0}$.
\end{assumption}
Since $\gamma_{\psi_i}(t)$ and $\rho_{\psi_i}^t(t)$ are affected by design choices, Assumption \ref{Assumption on the positive difference between gamma and observer uncertainties} is not restrictive in practice. 

Given Assumption~\ref{Assumption on the positive difference between gamma and observer uncertainties}, the following holds for $\Gamma_{\psi_i}(t)$.
\begin{lemma}
    \label{lemma on prescribed performance function}
    $\Gamma_{\psi_i}(t)$ is a prescribed performance function as per Definition \ref{definition of prescribed performance function}.
\end{lemma}
\begin{pf}
    $\gamma_{\psi_i}(t)$ and $\rho_{\psi_i}^t(t)$ are prescribed performance functions. Thus, validity of (i)–(iii) in Definition~\ref{definition of prescribed performance function} can be established for $\Gamma_{\psi_i}(t)$ as follows. (i) $\gamma_{\psi_i}(t), \rho_{\psi_i}^t(t) \in \mathcal{C}^1$. Thus, $\Gamma_{\psi_i}(t) \in \mathcal{C}^1$. (ii) From Definition~\ref{definition of prescribed performance function}, $0 \leq \gamma_{\psi_i}(t) \leq \overline{\gamma}_{\psi_i}$ and $0 \leq \rho_{\psi_i}^t(t) \leq \overline{\rho}_{\psi_i}^t$ hold for some bounded $\overline{\gamma}_{\psi_i}$ and $\overline{\rho}_{\psi_i}^t \in \mathbb{R}_{> 0}$. Hence, under Assumption \ref{Assumption on the positive difference between gamma and observer uncertainties}, $0<\Gamma_{\psi_i}(t) < \overline{\gamma}_{\psi_i}$. (iii) $|\dot{\gamma}_{\psi_i}(t)|\leq \dot{\overline{\gamma}}_{\psi_i}$ and $|\dot{\rho}_{\psi_i}^t|\leq \dot{\overline{\rho}}_{\psi_i}^t$ hold for some $\dot{\overline{\gamma}}_{\psi_i}<\infty$ and $\dot{\overline{\rho}}_{\psi_i}^t < \infty$. Thus $|\dot{\Gamma}^{\psi_i}(t)| \leq |\dot{\gamma}_{\psi_i}(t)| + |\dot{\rho}_{\psi_i}^t|\leq \dot{\overline{\gamma}}_{\psi_i} + \dot{\overline{\rho}}_{\psi_i}^t$. {\hspace*{\fill}$\Box$\vspace{1ex}}
\end{pf}
\begin{remark}
    For collaborative tasks $\phi_i$ for which $\Neigh{i}{T} \setminus \Neigh{i}{C} = \emptyset$, i.e., tasks involving only communicating agents, $\hat{\vect{x}}_{\phi_i} \equiv \vect{x}_{\phi_i}$ holds. Consequently, $\rho_{\psi_i}^t(t) = 0$ for all $t$ and condition \eqref{eq: inequality on robustness to guarantee task satisfaction, observer not in design stage of gamma} reduces to \eqref{eq: robustness inequality no estimations ivolved}.
\end{remark}
To handle tasks $\phi_i$ of the form \eqref{eq: definition of temporal formula}, where the \mbox{non-temporal} formula $\psi_i$ is obtained as the conjunction of $p_i$ formulas, i.e., $\psi_i = \land_{j = 1}^{p_i} \psi_{i,j}$, we exploit the following properties of the always ($G$) and eventually ($F$) temporal operators:
\begin{equation}
    \label{eq: always operator properties}
    \begin{split}
        G_{[a,b]}\psi_i &= \land_{j = 1}^{p_i} G_{[a,b]} \psi_{i,j}, \\ F_{[a,b]}\psi_i &= G_{[\tau,\tau]} \psi_i \text{ with } \tau \in [a,b].
    \end{split}
\end{equation}
Given the validity of \eqref{eq: always operator properties}, each $\gamma_{\psi_{i,j}}(t)$ can be designed and tuned individually to obtain $\Gamma_{\psi_{i,j}}(t)$ for all $j \in \{1,\dots, p_i\}$. The satisfaction of $\psi_{i}$ is then enforced by constraining $ - \Gamma_{\psi_{i}}(t) < \rho^{\psi_i}(\hat{\vect{x}}_{\phi_i}) - \rho^{\max}_{\psi_i} < 0$ as in \eqref{eq: inequality on robustness to guarantee task satisfaction, observer not in design stage of gamma}, with $\Gamma_{\psi_{i}}(t) = \min_{j \in \{1, \dots, p_i\}} \Gamma_{\psi_{i,j}}(t) \geq -\frac{1}{\eta} \ln(\sum_{j = 1}^{p_i}\exp(-\eta \Gamma_{\psi_{i,j}}(t))$ and robustness $\rho^{\psi_i}(\hat{\vect{x}}_{\phi_i}) = \min_{j \in \{1, \dots, p_i\}} \rho^{\psi_{i,j}}(\hat{\vect{x}}_{\phi_i}) \geq -\frac{1}{\eta} \ln(\sum_{j = 1}^{p_i}\exp(-\eta \rho^{\psi_{i,j}}(\hat{\vect{x}}_{\phi_i}))$.

Based on \eqref{eq: inequality on robustness to guarantee task satisfaction, observer not in design stage of gamma}, Problem~\ref{Problem: first problem formulation} can be reformulated as:
\begin{problem}
    \label{Problem: second problem formulation}
    Consider a heterogeneous MAS \eqref{eq: agent's dynamic} with a \mbox{cluster-induced} graph $\mathcal{G}'$ as in Definition~\ref{def: cluster induced graph}, where each agent employs the $k$-hop PPSO \eqref{observer dynamics}. Design a control law $u_i$ that guarantees \eqref{eq: inequality on robustness to guarantee task satisfaction, observer not in design stage of gamma} for all $i \in \mathcal{V}$ and all $t \in \mathbb{R}_{\geq 0}$.
\end{problem}

\section{Decentralized Observer-based Controller}\label{Section: Decentralized controller}
Inspired by \cite{LINDEMANN2021100973,10918825}, this section leverages the concept of Prescribed Performance Control (PPC) \citep{4639441} to address Problem \ref{Problem: second problem formulation}.
% It is subsequently shown that the satisfaction of the local specifications guarantees the fulfillment of the global task.

%\subsection{Local controller design}
Under Assumption~ \ref{Assumption on the positive difference between gamma and observer uncertainties}, define the normalized error $e^{\psi_i}(\hat{\vect{x}}_{\phi_i}, t)$ as $e^{\psi_i}(\hat{\vect{x}}_{\phi_i}, t) = \Gamma_{\psi_i}^{-1}(t) (\rho^{\psi_i}(\hat{\vect{x}}_{\phi_i}) - \rho^{\max}_{\psi_i})$.
%Based on $e^{\psi_i}(\hat{\vect{x}}_{\phi_i}, t)$, the transformed error is defined as
The transformed error $\epsilon^{\psi_i}(\hat{\vect{x}}_{\phi_i}, t)$ is then defined as
\begin{equation}
    \label{eq: transformed erro definition}
    \begin{split}
    \epsilon^{\psi_i}(\hat{\vect{x}}_{\phi_i}, t) &= T_{\psi_i}(e^{\psi_i}(\hat{\vect{x}}_{\phi_i}, t)),
    \end{split}
\end{equation}
where $T_{\psi_i}: (-1,0) \rightarrow \mathbb{R}$ is a strictly increasing transformation given by $T_{\psi_i}(e^{\psi_i}(\hat{\vect{x}}_{\phi_i}, t)) = \ln\left(-\frac{e^{\psi_i}(\hat{\vect{x}}_{\phi_i}, t)+1}{e^{\psi_i}(\hat{\vect{x}}_{\phi_i}, t)}\right)$.
Furthermore, denote by $J_{\psi_i}(e^{\psi_i})= \frac{\partial{T_{\psi_i}}}{\partial e^{\psi_i}} = - \frac{1}{e^{\psi_i}(e^{\psi_i}+1)}$ the Jacobian of $T_{\psi_i}$. 
%Note that $J_{\psi_i}(e_{\psi_i}) > 0$ by definition.
The main idea of PPC is to design control policies $u_i$ that keep $\epsilon^{\psi_i}(\hat{\vect{x}}_{\phi_i}, t)$ bounded, thereby ensuring satisfaction of \eqref{eq: inequality on robustness to guarantee task satisfaction, observer not in design stage of gamma}. To this end, consider the following assumption on $\rho^{\psi_i}$.

\begin{assumption}
    \label{Assumption on concavity and weelposeness of the robustness function}
    Each formula $\psi_i$ is such that: (i) $\rho^{\psi_i}: \mathbb{R}^{N_i^T} \rightarrow \mathbb{R}$ is concave with  $\frac{\partial \rho^{\psi_i}(\vect{x}_{\phi_i})}{\partial x_{i}} = 0_n$ only at the global maximum (ii) $\rho^{\psi_i}$ is continuously differentiable with respect to $\vect{x}_{\phi_i}$ or smooth almost everywhere except from the global maximum and (iii) the formula is \mbox{well-posed} in the sense that for all $C \in \mathbb{R}$ there exists $0 \leq \Bar{C} < \infty$ such that for all $\vect{x}_{\phi_i}$ with $\rho^{\psi_i}(\vect{x}_{\phi_i}) \geq C$, $\norm{\vect{x}_{\phi_i}}\leq \Bar{C}$ holds.
\end{assumption}

% \red{Since linear functions and those associated with containment or formation tasks are concave, condition (i) of Assumption~\ref{Assumption on concavity and weelposeness of the robustness function} is not restrictive.
% Moreover, since under proper task allocation and convergent observer adding $\psi^{\text{Bound}}_i := (\norm{\Bar{\vect{x}}_{\phi_i}}<\Bar{C})$ to each $\psi_i$ guarantees boundedness of all the estimates of $\Bar{\vect{x}}_{\phi_i}$, it is always possible to impose (iii) by combining $\phi_i$ with $\psi^{\text{Bound}}_i$, where $\Bar{C}$ is selected sufficiently large to preserve feasibility. Therefore, (iii) of Assumption~\ref{Assumption on concavity and weelposeness of the robustness function} is also not stringent.}
\begin{remark}
    \label{Remark: weel-posenes of robustness}
    Since linear functions and those modeling containment or formation tasks are concave, condition (i) of Assumption~\ref{Assumption on concavity and weelposeness of the robustness function} is not restrictive.
    Moreover, each $\phi_i$ can be augmented with $\psi^{\text{Bound}}_i := (\norm{\vect{x}_{\phi_i}}<\Bar{C}_i)$, where choosing $\Bar{C}_i$ sufficiently large to preserve feasibility ensures that condition (iii) of Assumption~\ref{Assumption on concavity and weelposeness of the robustness function} is satisfied under proper task assignment and a convergent observer. For \mbox{collaborative} tasks, only $\Bar{\vect{x}}_{\phi_i}$ can be directly bounded by $\psi^{\text{Bound}}_i$. Nevertheless, as shown in Theorem~\ref{Theorem on task satisfacion}, under Assumption~\ref{Assumption on Task graph} and a convergent observer, it suffices to add $\psi^{\text{Bound}}_j := (\norm{\Bar{\vect{x}}_{\phi_j}}<\Bar{C}_j)$ to each $\psi_j$ that do not satisfy (iii) in order to guarantee boundedness of $\norm{\hat{\vect{x}}_{\phi_i}}$, and thus ensure that condition (iii) holds.
\end{remark}
Define the global optimum as $\rho^{\text{opt}}_{\psi_i} := \text{sup}_{\vect{x}_{\phi_i}} \rho^{\psi_i}(\vect{x}_{\phi_i})$. Given the concavity of $\rho^{\psi_i}$ from Assumption \ref{Assumption on concavity and weelposeness of the robustness function}, to avoid $\frac{\partial\rho^{\psi_i}(\vect{x}_{\phi_i})}{\partial x_{i}} = 0_n$ along the transient, and thus feasibility issue, it suffices to tune $\rho^{\max}_{\psi_i}$ so that $0 < \rho^{\max}_{\psi_i} <\rho^{\text{opt}}_{\psi_i}$.

For given observer prescribed performances $\delta_{N^i_j}^i(t)$, with $N^i_j \in \Nhop{i}{k}$, ensuring the satisfiability of $\phi_i$ regardless of $\hat{\vect{x}}_{\phi_i}$, for every $i \in \mathcal{V}$, requires the following.
% Denote by $\rho^{\text{opt}}_{\psi_i}$ the global maximum of the robustness $\rho^{\psi_i}(\vect{x}_{\phi_i})$, i.e., $\rho^{\text{opt}}_{\psi_i} := \text{sup}_{\vect{x}_{\phi_i}} \rho^{\psi_i}(\vect{x}_{\phi_i})$. Given the concavity from Assumption~\ref{Assumption on concavity and weelposeness of the robustness function}, the optimum point is unique and $\frac{\partial \rho^{\psi_{i}}(\vect{x}_{\phi_i})}{\partial x_{i}} = 0$ holds only when $\rho^{\psi_i}(\vect{x}_{\phi_i}) = \rho^{\text{opt}}_{\psi_i}$. Thus, to guarantee $\frac{\partial\rho^{\psi_i}(\vect{x}_{\phi_i})}{\partial x_{i}} \neq 0_n$ along the transients, it suffices to tune $\rho^{\max}_{\psi_i}$ so that $0 < \rho^{\max}_{\psi_i} <\rho^{\text{opt}}_{\psi_i}$.
% For given observer prescribed performances $\delta_{N^i_j}^i(t)$, with $N^i_j \in \Nhop{i}{k}$, ensuring the satisfiability of $\phi_i$ regardless of $\hat{\vect{x}}_{\phi_i}$, for every $i \in \mathcal{V}$, requires the following.
\begin{assumption}
    \label{Assumption: optimal robustness greater than zero}
    $\rho^{\max}_{\psi_i} - \max_{\tau \in T_{\text{wdw}}} \rho_{\psi_i}^t(\tau) > 0$ for all $i \in \mathcal{V}$, where $T_{\text{wdw}}$ represents the time window over which $\phi_i$ must be satisfied.
\end{assumption}

Assumption~\ref{Assumption: optimal robustness greater than zero} represents a feasibility condition for given $\delta_{N^i_j}^i(t)$ and $\rho^{\max}_{\psi_i}$. However, since $\rho_{\psi_i}^t(t)$ depends on $\delta_{N^i_j}^i(t)$, which is a design choice, Assumption~\ref{Assumption: optimal robustness greater than zero} can be satisfied whenever $\rho_{\text{opt}}^{\psi_i} > 0$ by appropriate design of $\delta_{N^i_j}^i(t)$ and $\rho^{\max}_{\psi_i}$. Thus, the assumption is not restrictive.
For always operators $G_{[a,b]} \phi_i$, $T_{\text{wdw}} = [a,b]$, while for eventually operators $F_{[a,b]} \phi_i$, $T_{\text{wdw}} = \{t^*\}$, where $t^* \in [a,b]$ is the time instant of task satisfaction, as per Section~\ref{Robustness adaptation} and \cite{LINDEMANN2021100973}. 
To present our main result, an additional assumption is required.
\begin{assumption}
    \label{Assumption on initialization}
    %The states and state estimates are initialized so that 
    $-\Gamma_{\psi_i}(0) < {\rho}^{\psi_i}(\hat{\vect{x}}_{\phi_i}(0)) - \rho^{\max}_{\psi_i} < 0$ and $|\xi^{i}_{N^i_j}(0)| < \rho_{N_j^i}^i(0)$ hold for all $i \in \mathcal{V}$ and $N_j^i \in \Nhop{i}{k}$.
\end{assumption}
Assumption \ref{Assumption on initialization} involves only design parameters and initialization quantities. Therefore, it's not restrictive.
With the defined notation and assumptions, we are now ready to state the main result of this work.

\begin{thm}
    \label{Theorem on task satisfacion}
    Consider a MAS \eqref{eq: agent's dynamic} with communication graph $\mathcal{G}_C$, and task graph $\mathcal{G}_T$ induced by a global task $\phi =\land_{i=1}^N \phi_i$. Suppose each agent runs the $k$-hop PPSO in \eqref{observer dynamics} with $k$ selected according to \eqref{selection of the parameter k}. Consider the agents belonging to the cluster $\mathcal{C}_l$, subject to the task $\phi^{c}_l = \land_{l_i \in \mathcal{V}_l} \phi_{l_i}$. Then, $\phi^c_l$ is satisfied, for all $l \in \mathcal{L}$, if Assumptions \ref{Assumption on neighbors}, \ref{Assumption on Task graph}, \ref{Assumption on the non-existance of a task not in communication inside a cluster}, \ref{assumption on robustness lower bound}, \ref{Assumption on the positive difference between gamma and observer uncertainties}, \ref{Assumption on concavity and weelposeness of the robustness function}, \ref{Assumption: optimal robustness greater than zero}, \ref{Assumption on initialization} hold and each agent $l_i\in \mathcal{V}_l$ applies the decentralized controller
    \begin{equation}
        \label{eq: agent's input}
        u_{l_i} = - g^{\top}_{l_i}(x_{l_i}(t))\sum_{l_j\in \mathcal{V}_l} \frac{\partial {\rho}^{\psi_{l_j}}(\hat{\vect{x}}_{\phi_{l_j}})}{\partial x_{{l_i}}} \Gamma^{-1}_{\psi_{l_j}}  J_{\psi_{l_j}} \epsilon^{\psi_{l_j}},
    \end{equation}
    where $\Gamma_{\psi_{l_j}}$ is defined as in \eqref{eq: inequality on robustness to guarantee task satisfaction, observer not in design stage of gamma}, and $J_{\psi_{l_j}}$, $\epsilon^{\psi_{l_j}}$ as in \eqref{eq: transformed erro definition}.
\end{thm} 
\begin{pf}
    The proof is structured in two steps. In Step~A, task satisfaction is established for clusters $\mathcal{C}_l \in \mathcal{G}'$ corresponding to leaf nodes in $\mathcal{G}'$. In Step B, observer convergence and Assumption \ref{Assumption on Task graph} are used to prove task satisfaction for all remaining clusters.
    
    Let's start by computing the transformed error dynamics. For notational simplicity, $\hat{\rho}^{\psi_i} := \rho^{\psi_i}(\hat{\vect{x}}_{\phi_i})$ and $\rho^{\psi_i} := \rho^{\psi_i}(\vect{x}_{\phi_i})$ in the following.  
    Given the definition in \eqref{eq: transformed erro definition}, $\dot{\epsilon}^{\psi_i}(\hat{\vect{x}}_{\phi_i}, t) = -\frac{1}{e^{\psi_i}(1+e^{\psi_i})}\dot{e}^{\psi_i}$ with $\dot{e}^{\psi_i} =  \Gamma_{\psi_i}^{-1}\Bigl(\frac{\partial \hat{\rho}^{\psi_i}}{\partial \Bar{\vect{x}}_{\phi_i}}^{\top} \dot{\Bar{\vect{x}}}_{\phi_i} + \frac{\partial \hat{\rho}^{\psi_i}}{\partial \estimate{\vect{x}}{i}{{\phi_i}}}^{\top}  \dot{\Hat{\vect{x}}}^{i}_{\phi_i} - \dot{\Gamma}_{\psi_i}(t)e^{\psi_i}\Bigr)$, 
    % Thus,
    % \begin{equation}
    %     \label{eq: transformed normalized error dynamics}
    %     \begin{split}
    %         \dot{\epsilon}^{\psi_i} = & J_{\psi_i}\Gamma^{-1}_{\psi_i}\biggl(\frac{\partial \hat{\rho}^{\psi_i}}{\partial \Bar{\vect{x}}_{\phi_i}}^{\top}  \dot{\Bar{\vect{x}}}_{\phi_i}  + \frac{\partial \hat{\rho}^{\psi_i}}{\partial \estimate{\vect{x}}{i}{{\phi_i}}}^{\top}  \dot{\Hat{\vect{x}}}^{i}_{\phi_i}  -\dot{\Gamma}_{\psi_i}e^{\psi_i}\biggr),
    %     \end{split}
    % \end{equation}
    where $\frac{\partial \hat{\rho}^{\psi_i}}{\partial \Bar{\vect{x}}_{\phi_i}} \in \mathbb{R}^{N^{TC}_{i}}$ and $\frac{\partial \hat{\rho}^{\psi_i}}{\partial \estimate{\vect{x}}{i}{{\phi_i}}} \in \mathbb{R}^{N^{T\setminus C}_{i}}$ are vectors containing, respectively, the derivatives of $\hat{\rho}^{\psi_i}$ with respect to the states of agents $j\in \Neigh{i}{T} \cap \Neigh{i}{C}$ and $j\in \Neigh{i}{T} \setminus \Neigh{i}{C}$, and where $\dot{\Bar{\vect{x}}}_{\phi_i}$ and $\dot{\Hat{\vect{x}}}^{i}_{\phi_i}$ denote the dynamics of $x_j$, $j\in \Neigh{i}{T} \cap \Neigh{i}{C}$, and the one of $\hat{x}^i_j$, $j\in \Neigh{i}{T} \setminus \Neigh{i}{C}$, when the estimate is performed by $i$.  
    As introduced in Section \ref{Section: collaborative and non-collaborative tasks}, $\Bar{\vect{x}}_{\phi_i}$ is the vector containing the true state of agents  $j\in \Neigh{i}{T} \cap \Neigh{i}{C}$. Thus, $\dot{\epsilon}^{\psi_i}$ can be rewritten as:
    \begin{equation}
        \label{eq: transformed normalized error dynamics}
        \begin{split}
            \dot{\epsilon}^{\psi_i} = & J_{\psi_i}\Gamma^{-1}_{\psi_i}\biggl(\frac{\partial \hat{\rho}^{\psi_i}}{\partial {\vect{x}}_{\phi_i}}^{\top}  \dot{{\vect{x}}}_{\phi_i}  + \frac{\partial \hat{\rho}^{\psi_i}}{\partial \estimate{\vect{x}}{i}{{\phi_i}}}^{\top}  \dot{\Hat{\vect{x}}}^{i}_{\phi_i}  -\dot{\Gamma}_{\psi_i}e^{\psi_i}\biggr),
        \end{split}
    \end{equation}
    where $\frac{\partial \hat{\rho}^{\psi_i}}{\partial {\vect{x}}_{\phi_i}}$ is the derivative of $\hat{\rho}^{\psi_i}$ with respect to the state of the agents in cluster $\mathcal{C}_l  = (\mathcal{V}_l, \mathcal{E}_l)$, with $i \in \mathcal{V}_l$. Note that, from condition (i) in Assumption \ref{Assumption on concavity and weelposeness of the robustness function}, and since $0 < \rho^{\max}_{\psi_i} <\rho^{\text{opt}}_{\psi_i}$ holds for all $i \in \mathcal{V}$, the only components of $\frac{\partial \hat{\rho}^{\psi_i}}{\partial {\vect{x}}_{\phi_i}}$ equal to zero are those associated to agents $j \notin \Neigh{i}{T}$ and $j \in \Neigh{i}{T} \setminus \Neigh{i}{C}$.
    By stacking the transformed errors of cluster $\mathcal{C}_l$, i.e., $\epsilon^{\psi_{l_i}}$ for all $l_i \in \mathcal{V}_l$, we define the transformed error vector  $\vect{\epsilon}^{\phi^c_l} := [\epsilon^{\psi_{l_1}}, \dots, \epsilon^{\psi_{l_{v_l}}}]^{\top}$  with dynamics
    \begin{equation}
        \hspace{-0.2cm}
        \label{eq: cluster transformed normalized error dynamics}
        \dot{ \vect{\epsilon}}^{\phi^c_l} = \vect{J}_{\phi^c_l} \vect{\Gamma}^{-1}_{\phi^c_l}  \{ \vect{\Lambda}_{\vect{x}_{\phi^c_l}}^{\top} \dot{\vect{x}}_{\phi^c_l} + \vect{\Lambda}_{\Hat{\vect{x}}_{\phi^c_l}}^{\top}\dot{\Hat{\vect{x}}}_{\phi^c_l} - \dot{\vect{\Gamma}}_{\phi^c_l} \vect{e}^{\phi^c_l}\},
    \end{equation}
    where $\vect{x}_{\phi^c_l} = [x^\top_{l_1}, \dots,  x^\top_{l_{v_l}}]^{\top}$ and $\Hat{\vect{x}}_{\phi^c_l} = [\estimate{\vect{x}}{{l_1 \top}}{{\phi_{l_1}}}, \dots , \estimate{\vect{x}}{{l_{v_l}} \top}{{\phi_{l_{v_l}}}}]^\top $ are defined as in Section \ref{Section: collaborative and non-collaborative tasks}, $\vect{J}_{\phi^c_l} = \text{diag}(J_{\psi_{l_1}}, \dots, J_{\psi_{l_{v_l}}} )$, $\vect{\Gamma}_{\phi^c_l} = \text{diag}(\Gamma_{\psi_{l_1}}, \dots, \Gamma_{\psi_{l_{v_l}}})$,  $\vect{\Lambda}_{\vect{x}_{\phi^c_l}} = \Bigl[ \frac{\partial \hat{\rho}^{\psi_{l_1}}}{\partial \vect{x}_{\phi^c_l}} \ \dots \  \frac{\partial \hat{\rho}^{\psi_{l_{v_l}}}}{\partial \vect{x}_{\phi^c_l}} \Bigr]$, 
      $\vect{\Lambda}_{\Hat{\vect{x}}_{\phi^c_l}} = \text{diag}\Bigl( \frac{\partial \hat{\rho}^{\psi_{l_1}}}{\partial \estimate{\vect{x}}{{l_1}}{{\phi_{l_1}}}}, \dots , \frac{\partial \hat{\rho}^{\psi_{l_{v_l}}}}{\partial \estimate{\vect{x}}{l_{v_l}}{{\phi_{l_{v_l}}}}}\Bigr)$ and $\vect{e}^{\phi^c_l} = [e^{\psi_{l_1}}, \dots,  e^{\psi_{l_{v_l}}}]^{\top}$.
    Given the definition of $\vect{x}_{\phi^c_l}$ from Section \ref{Section: collaborative and non-collaborative tasks}, its dynamics can be written as:
    \begin{equation}
        \label{eq: cluster state dynamcis}
        \dot{\vect{x}}_{\phi^c_l} = \vect{f}_{\phi^c_l}(\vect{x}_{\phi^c_l}) +\vect{g}_{\phi^c_l}(\vect{x}_{\phi^c_l})\vect{u}_{\phi^c_l} + \vect{w}_{\phi^c_l}(\vect{x},t),
    \end{equation}
    where $\vect{f}_{\phi^c_l}(\vect{x}_{\phi^c_l}) = [f_{l_1}(x_{l_1}), \dots ,f_{l_{v_l}}(x_{l_{v_l}})]^\top$, $\vect{g}_{\phi^c_l}(\vect{x}_{\phi^c_l}) = \text{diag}(g_{l_1}(x_{l_1}), \dots ,g_{l_{v_l}}(x_{l_{v_l}}))$,  $\vect{u}_{\phi^c_l} = [u_{l_1}, \dots, u_{l_{v_l}}]^\top$ and $\vect{w}_{\phi^c_l}(\vect{x},t) = [w_{l_1}(\vect{x},t), \dots, w_{l_{v_l}}(\vect{x},t)]^\top$.

   With \eqref{eq: cluster transformed normalized error dynamics} and \eqref{eq: cluster state dynamcis}, the two steps of the proof can be addressed as follows.

    \textbf{Step~A (leaf clusters):} As discussed in Remark \ref{remark on the acyclicity of the cluster induced graph}, $\mathcal{G}'$ is a directed acyclic graph. Therefore, there exists at least one leaf cluster $\mathcal{C}_l = (\mathcal{V}_l, \mathcal{E}_l)$, $l \in \mathcal{L}$, with no outgoing edges. Consequently, from Definition \ref{def: cluster induced graph} and Assumption~\ref{Assumption on the non-existance of a task not in communication inside a cluster}, all $\phi_{l_i}$, with $l_i\in \mathcal{V}_l$, involve only agents in $\mathcal{V}_l$, and satisfaction of
    $\phi_l^c$ is independent of estimations. As a result, $\hat{\rho}^{\psi_i} \equiv \rho^{\psi_i}$ as per Remark \ref{Remark on robustness equivalence under full state knowledge}, $\frac{\partial \hat{\rho}^{\psi_i}}{\partial \estimate{\vect{x}}{i}{{\phi_i}}}^{\top} = 0$, and \eqref{eq: cluster transformed normalized error dynamics} boils down to $\dot{\vect{\epsilon}}^{\phi^c_l} = \vect{J}_{\phi^c_l} \vect{\Gamma}^{-1}_{\phi^c_l}  \{ \vect{\Lambda}_{\vect{x}_{\phi^c_l}}^{\top} \dot{\vect{x}}_{\phi^c_l}- \dot{\vect{\Gamma}}_{\phi^c_l} \vect{e}^{\phi^c_l}\}$,
    % As a result, $\vect{x}_{\phi_{l_i}} \equiv \Bar{\vect{x}}_{\phi_{l_i}} \equiv \hat{\vect{x}}_{\phi_{l_i}}$, $\frac{\partial \hat{\rho}^{\psi_i}}{\partial \estimate{\vect{x}}{i}{{\phi_i}}}^{\top} = 0$, and \eqref{eq: cluster transformed normalized error dynamics} boils down to:
    % \begin{equation}
    %     \label{eq: leaf cluster transformed error}
    %     \dot{\vect{\epsilon}}^{\phi^c_l} = \vect{J}_{\phi^c_l} \vect{\Gamma}^{-1}_{\phi^c_l}  \{ \vect{\Lambda}_{\vect{x}_{\phi^c_l}}^{\top} \dot{\vect{x}}_{\phi^c_l}- \dot{\vect{\Gamma}}_{\phi^c_l} \vect{e}^{\phi^c_l}\},
    % \end{equation}
    where $\vect{\Lambda}_{\vect{x}_{\phi^c_l}} = \Bigl[ \frac{\partial {\rho}^{\psi_{l_1}}}{\partial \vect{x}_{\phi^c_l}} \ \dots \  \frac{\partial {\rho}^{\psi_{l_{v_l}}}}{\partial \vect{x}_{\phi^c_l}} \Bigr]$ and $\vect{\Gamma}_{\phi^c_l} = \text{diag}(\gamma_{\psi_{l_i}}, \dots,\gamma_{\psi_{l_{v_l}}})$.
    Consider now the candidate Lyapunov function $V = \frac{1}{2} \vect{\epsilon}^{\phi^c_l\top} \vect{\epsilon}^{\phi^c_l}$. By differentiating $V$ and by replacing $\dot{\vect{\epsilon}}^{\phi^c_l}$, $\dot{V}$ becomes $ \dot{V} =  \vect{\epsilon}^{\phi^c_l\top} \vect{J}_{\phi^c_l} \vect{\Gamma}^{-1}_{\phi^c_l}  \{ \vect{\Lambda}_{\vect{x}_{\phi^c_l}}^{\top} \dot{\vect{x}}_{\phi^c_l}- \dot{\vect{\Gamma}}_{\phi^c_l} \vect{e}^{\phi^c_l}\}$. Then, by plugging \eqref{eq: cluster state dynamcis}, $ \dot{V} =  \vect{\epsilon}^{\phi^c_l\top} \vect{J}_{\phi^c_l} \vect{\Gamma}^{-1}_{\phi^c_l}  \{ \vect{\Lambda}_{\vect{x}_{\phi^c_l}}^{\top} [\vect{f}_{\phi^c_l}(\vect{x}_{\phi^c_l}) +  \vect{w}_{\phi^c_l}(\vect{x},t) +\vect{g}_{\phi^c_l}(\vect{x}_{\phi^c_l})\vect{u}_{\phi^c_l} ]- \dot{\vect{\Gamma}}_{\phi^c_l} \vect{e}^{\phi^c_l}\}$ holds.
    By stacking $u_{l_i}$ as per \eqref{eq: agent's input} for all $l_i \in \mathcal{V}_l$, the input vector is $\vect{u}_{\phi^c_l} = -\vect{g}^{\top}_{\phi^c_l}(\vect{x}_{\phi^c_l})  \vect{\Lambda}_{\vect{x}_{\phi^c_l}}^{\top} \vect{\Gamma}^{-1}_{\phi^c_l} \vect{J}_{\phi^c_l}\vect{\epsilon}^{\phi^c_l}$ and 
     \begin{equation} 
        \label{eq: Lyapunov function expression}
        \begin{split}
            \dot{V} =& - \vect{\epsilon}^{\phi^c_l\top} \vect{J}_{\phi^c_l} \vect{\Gamma}^{-1}_{\phi^c_l} \vect{\Lambda}_{\vect{x}_{\phi^c_l}}^{\top} \Theta_{\phi^c_l}  \vect{\Lambda}_{\vect{x}_{\phi^c_l}} \vect{\Gamma}^{-1}_{\phi^c_l} \vect{J}_{\phi^c_l}\vect{\epsilon}^{\phi^c_l} \\  &+ \vect{\epsilon}^{\phi^c_l\top} \vect{J}_{\phi^c_l} \vect{\Gamma}^{-1}_{\phi^c_l} \{- \dot{\vect{\Gamma}}_{\phi^c_l} \vect{e}^{\phi^c_l} + \vect{\Lambda}_{\vect{x}_{\phi^c_l}}^{\top} [\vect{f}_{\phi^c_l}+ \vect{w}_{\phi^c_l}]\},
        \end{split}
    \end{equation}
    where $\Theta_{\phi^c_l} = \vect{g}_{\phi^c_l}(\vect{x}_{\phi^c_l}) \vect{g}^\top_{\phi^c_l}(\vect{x}_{\phi^c_l})$ is a positive definite block diagonal matrix by condition (ii) of Assumption~\ref{Assumption on existence of a solution}.
     Given the acyclicity of $\mathcal{G}_T$ from Assumption \ref{Assumption on Task graph}, $\vect{\Lambda}_{\vect{x}_{\phi^c_l}}$ can always be converted, by a proper permutation, into a block \mbox{upper-triangular} matrix. Furthermore, since $\rho^{\max}_{\psi_{l_i }} <\rho^{\text{opt}}_{\psi_{l_i }}$ holds by design, and by (i) of Assumption~\ref{Assumption on concavity and weelposeness of the robustness function}, $\frac{\partial {\rho}^{\psi_{l_i}}}{\partial x_{l_i}}$ are guaranteed to be non-zero for all $l_i \in \mathcal{V}_l$ whenever ${\rho}^{\psi_{l_i}} < \rho^{\text{opt}}_{\psi_{l_i }}$,  the elements on the main diagonal of $\vect{\Lambda}_{\vect{x}_{\phi^c_l}}$ are nonzero. As a result, $\text{rank}(\vect{\Lambda}_{\vect{x}_{\phi^c_l}}) = v_l$ and $\vect{\Lambda}_{\vect{x}_{\phi^c_l}}^{\top} \Theta_{\phi^c_l} \vect{\Lambda}_{\vect{x}_{\phi^c_l}}$ is positive definite. Since $\vect{J}_{\phi^c_l}$ and $\vect{\Gamma}^{-1}_{\phi^c_l} $ are diagonal matrices with positive diagonal entries $J_{\psi_{l_i}}(e^{\psi_{l_i}})$ and $\gamma_{\psi_{l_i}}(t)$, $- \vect{\epsilon}^{\phi^c_l\top } \vect{J}_{\phi^c_l} \vect{\Gamma}^{-1}_{\phi^c_l}\vect{\Lambda}_{\vect{x}_{\phi^c_l}}^{\top} \Theta_{\phi^c_l} \vect{\Lambda}_{\vect{x}_{\phi^c_l}} \vect{\Gamma}^{-1}_{\phi^c_l} \vect{J}_{\phi^c_l}\vect{\epsilon}^{\phi^c_l} <  - \alpha_{\rho} \alpha_{J} \vect{\epsilon}^{\phi^c_l \top}\vect{\epsilon}^{\phi^c_l} $ holds with $\alpha_{\rho} = \lambda_{\text{min}}(\vect{\Lambda}_{\vect{x}_{\phi^c_l}}^{\top} \Theta_{\phi^c_l} \vect{\Lambda}_{\vect{x}_{\phi^c_l}}) \in \mathbb{R}$ and $\alpha_J = \min_{l_i \in \mathcal{V}_l} \Bigl\{\min_{t \in \mathbb{R}_{\geq 0},e^{\psi_{l_i}}\in (-1,0) } \gamma_{\psi_{l_i}}^{-2}(t) J^2_{\psi_{l_i}}(e^{\psi_{l_i}}) \Bigr\} \in \mathbb{R}$. By adding and subtracting $\zeta \norm{\vect{\Gamma}^{-1}_{\phi^c_l} \vect{J}_{\phi^c_l}\vect{\epsilon}^{\phi^c_l}}^2$ for some $0 < \zeta < \alpha_\rho$, \eqref{eq: Lyapunov function expression} can be rewritten as $\dot{V} \leq - (\alpha_{\rho} - \zeta) \alpha_J \norm{\vect{\epsilon}^{\phi^c_l}}^2 + \vect{\epsilon}^{\phi^c_l \top} \vect{J}_{\phi^c_l} \vect{\Gamma}^{-1}_{\phi^c_l} b(t) - \zeta \norm{\vect{\Gamma}^{-1}_{\phi^c_l} \vect{J}_{\phi^c_l}\vect{\epsilon}^{\phi^c_l}}^2$, where $b(t) =  \Bigl\{- \dot{\vect{\Gamma}}_{\phi^c_l} \vect{e}^{\phi^c_l} + \vect{\Lambda}_{\vect{x}_{\phi^c_l}}^{\top} [\vect{f}_{\phi^c_l}+ \vect{w}_{\phi^c_l}]\Bigr\}$. By noticing that $\vect{\epsilon}^{\phi^c_l \top} \vect{J}_{\phi^c_l} \vect{\Gamma}^{-1}_{\phi^c_l} b(t) - \zeta \norm{\vect{\Gamma}^{-1}_{\phi^c_l} \vect{J}_{\phi^c_l}\vect{\epsilon}^{\phi^c_l}}^2$ resemble terms of the quadratic form $\norm{\sqrt{\zeta} \vect{\Gamma}^{-1}_{\phi^c_l} \vect{J}_{\phi^c_l}\vect{\epsilon}^{\phi^c_l} - \frac{1}{2\sqrt{\zeta}}b(t)}^2$, we can upper bound them by $\frac{1}{4\zeta}b^\top(t)b(t)$. As a result, \eqref{eq: Lyapunov function expression} satisfies
     \begin{equation}
        \label{eq: Lyapunov function final upper bound}
        \dot{V} \leq - \kappa V + \vect{b}(t),
    \end{equation}
    where $\kappa = 2(\alpha_{\rho} - \zeta) \alpha_J $ and $\vect{b}(t) = \frac{1}{4\zeta} (\norm{\vect{\Lambda}_{\vect{x}_{\phi^c_l}}}\norm{\vect{f}_{\phi^c_l}} + \norm{\vect{\Lambda}_{\vect{x}_{\phi^c_l}}}\norm{\vect{w}_{\phi^c_l}} + \norm{\dot{\vect{\Gamma}}_{\phi^c_l}  \vect{e}^{\phi^c_l}})^2$. To proceed, let's check whether $\vect{b}(t)$ admit an upper bound $\Bar{\vect{b}}(t)$. For this purpose, define the set $\mathcal{X}_{l_i}(t):=\{\vect{x}_{\phi_{l_i}} \in \mathbb{R}^{N^{TC}_{i}}| -1 < e^{\psi_{l_i}} < 0 \}$ as the one containing the states $\vect{x}_{\phi_{l_i}}$ satisfying task $\phi_{l_i}$ at time $t$. Since $\gamma_{\phi_{l_i}}(t)$ is a non-increasing function, $\mathcal{X}_{l_i}(t_2) \subseteq \mathcal{X}_{l_i}(t_1)$ holds for all $t_1 < t_2$, and $\mathcal{X}_{l_i}(0)$ collects all states $\vect{x}_{\phi_{l_i}}$ for which $e^{\psi_{l_i}} \in (-1,0)$ for all $t \in \mathbb{R}_{\geq 0}$. From condition (iii) of Assumption~\ref{Assumption on concavity and weelposeness of the robustness function}, $\mathcal{X}_{l_i}(0)$ is bounded for all $l_i \in \mathcal{V}_l$. Thus, since $f$ is Lipschitz continuous, $\norm{\vect{f}_{\phi^c_l}(\vect{x}_{\phi^c_l})}$ is bounded in $\mathcal{X}_{l_i}(0)$. Moreover, since $\dot{\vect{\Gamma}}_{\phi^c_l} \vect{e}^{\phi^c_l}$ is a column vector with $\dot{\gamma}_{\phi_{l_i}}e^{\psi_{l_i}}$ as entries, and $|\dot{\gamma}_{\phi_{l_i}}(t)e^{\psi_{l_i}}| < |\dot{\gamma}_{\phi_{l_i}}(0)|$ holds with $|\dot{\gamma}_{\phi_{l_i}}(0)|< \infty$ from $\gamma_{\phi_{l_i}}$ definition, also $\norm{\dot{\gamma}_{\phi_{l_i}}e^{\psi_{l_i}}}$ is bounded. 
    %Let's now analyze the boundedness of $\norm{\vect{\Lambda}_{\vect{x}_{\phi^c_l}}}$. To this end, consider that
    Note that ${\rho}^{\psi_{l_i}}$ is concave, continuously differentiable and \mbox{well-posed} according to Assumption~\ref{Assumption on concavity and weelposeness of the robustness function}. Thus, since every function of this kind has bounded derivative on bounded open sets contained in its domain, $\norm{\vect{\Lambda}_{\vect{x}_{\phi^c_l}}}$ is bounded on $\mathcal{X}_{l_i}(0)$.
    Finally, since $w_i(\vect{x},t)$ satisfies Assumption \ref{Assumption on existence of a solution}, $\norm{\vect{w}_{\phi^c_l}}$ is bounded and an upper bound $\Bar{\vect{b}}(t)$ on $\vect{b}(t)$ is guaranteed to exist for all $\vect{x}_{\phi_{l_i}} \in \mathcal{X}_{l_i}(0)$ and all $l_i \in \mathcal{V}_l$. To conclude Step~A, consider $S(\vect{e}^{\phi^c_l}) = 1 - \exp(- V(\vect{e}^{\phi^c_l}))$. From its definition: (i) $S(\vect{e}^{\phi^c_l}) \in (0,1)$ for all $\vect{e}^{\psi_{l_i}} \in \mathcal{D} $, with $\mathcal{D} = \bigtimes_{j=1}^{v_l}(-1,0)$, and (ii) $S(\vect{e}^{\phi^c_l}) \rightarrow 1$ as $\vect{e}^{\phi^c_l} \rightarrow \partial \mathcal{D} $. Thus, studying the boundedness of $\vect{\epsilon}^{\phi^c_l}$ through the one of $V$, reduces to proving that $S(\vect{e}^{\phi^c_l}) < 1$ holds for all $t$. By differentiating $S(\vect{e}^{\phi^c_l})$, and by replacing \eqref{eq: Lyapunov function final upper bound} and $V(\vect{e}^{\phi^c_l}) = - \ln (1-S(\vect{e}^{\phi^c_l}))$, the following inequality is obtained: $\dot{S}(\vect{e}^{\phi^c_l}) \leq - \kappa (1-S(\vect{e}^{\phi^c_l})) \Bigl(-\frac{1}{\kappa}\vect{b}(t) - \ln (1-S(\vect{e}^{\phi^c_l}))\Bigr)$.
    Since $\kappa$ and $1-S(\vect{e}^{\phi^c_l})$ are positive, they do not affect the sign of $ \dot{S}(\vect{e}^{\phi^c_l})$. Therefore, to verify whether $\dot{S}(\vect{e}^{\phi^c_l}) \leq 0 $, it suffices to study under which condition $-\frac{1}{\kappa}\vect{b}(t) - \ln (1-S(\vect{e}^{\phi^c_l}))\geq 0$ holds. Note that
 $-\frac{1}{\kappa}\vect{b}(t) - \ln (1-S(\vect{e}^{\phi^c_l}))\geq 0 $ is satisfied for all $\vect{e}^{\phi^c_l} \in \Omega^c_{\vect{e}}$, where $\Omega_{\vect{e}}=\bigl\{\vect{e}^{\phi^c_l} \in \mathcal{D}| S(\vect{e}^{\phi^c_l}) < 1- \exp{(-\frac{\bar{\vect{b}}(t)}{\kappa})} \bigr\}$. Furthermore, $-\frac{1}{\kappa}\vect{b}(t) - \ln (1-S(\vect{e}^{\phi^c_l})) = 0 $ holds for $\vect{e}^{\phi^c_l} \in  \partial\Omega_{\vect{e}}$. Consequently, $\dot{S}(\vect{e}^{\phi^c_l}) \leq 0$  for $\vect{e}^{\phi^c_l} \in \Omega^c_{\vect{e}}$, with $\dot{S}(\vect{e}^{\phi^c_l}) =0$ iff $\vect{e}^{\phi^c_l} \in \partial\Omega_{\vect{e}}$. 
    Since the initialization satisfies $\rho^{\max}_{\psi_{l_i}} -\gamma_{\psi_{l_i}}(0) < \rho^{\psi_{l_i}}(\Bar{\vect{x}}_{\phi_{l_i}}(0)) < \rho^{\max}_{\psi_{l_i}}$ for all $l_i \in \mathcal{V}_l$, it follows that $e^{\psi_{l_i}} \in (-1,0)$ and $S(\vect{e}^{\phi^c_l}(0)) < 1$. Moreover, since $\exp{(-\frac{\Bar{\vect{b}}(t)}{\kappa})} > 0$ by definition, $S(\vect{e}^{\phi^c_l})) < 1 $ is preserved for all $t$, independently of whether $\vect{e}^{\phi^c_l}$ is initialized inside or outside $\Omega_{\vect{e}}$. From $S(\vect{e}^{\phi^c_l})) < 1 $, boundedness of $V(\vect{e}^{\phi^c_l})$, and therefore of the transformed error $\vect{\epsilon}^{\phi^c_l}$, follows. As a result, ${\rho}^{\psi_{l_i}}$ satisfies \eqref{eq: inequality on robustness to guarantee task satisfaction, observer not in design stage of gamma} for all $l_i \in \mathcal{V}_l$, and tasks satisfaction is guaranteed for every leaf cluster.

    \textbf{Step B (non-leaf clusters):} Consider $\mathcal{C}_l$ that are not leaf clusters. For simplicity, focus on an \mbox{in-neighbor} of a leaf cluster, i.e., $\mathcal{C}_l$ such that $(\mathcal{C}_l, \mathcal{C}_j) \in \mathcal{E}'$ with $\mathcal{C}_j$ being a leaf cluster. From Step~A, $\psi_{j_q}$ is satisfied for all  $j_q \in \mathcal{V}_j$. Hence, $\epsilon^{\psi_{j_q}} < \infty$ for all $j_q \in \mathcal{V}_j$, and $u_{j_q} < \infty$ holds. Since $\vect{x}_{\phi_{j_q}} \in \mathcal{X}_{j_q}(t)$, with $\mathcal{X}_{j_q}(t)$ bounded, all trajectories of agents  $j_q \in \mathcal{V}_j$ are bounded. Thus, the \mbox{$k$-hop} PPSO  in \eqref{observer dynamics} guarantees $|\tilde{x}^{N^{j_q}_r}_{j_q}| < \delta^{N^{j_q}_r}_{j_q}$ for all $j_q\in \mathcal{V}_j$ and $N^{j_q}_r \in \Nhop{j_q}{k}$ \cite[Thm.~2]{zaccherini2025robustestimationcontrolheterogeneous}.
    Based on the observer guarantees, task satisfaction of non-leaf clusters is proven similarly to Step~A.
    Consider $V = \frac{1}{2} \vect{\epsilon}^{\phi^c_l\top} \vect{\epsilon}^{\phi^c_l}$. By means of \eqref{eq: agent's input}, \eqref{eq: cluster transformed normalized error dynamics} and \eqref{eq: cluster state dynamcis}, $\dot{V}$ can be written as
    \begin{equation} 
        \label{eq: Lyapunov function expression in inner clusters}
        \begin{split}
        \hspace{-0.1cm}
            \dot{V} =& \vect{\epsilon}^{\phi^c_l\top} \vect{J}_{\phi^c_l} \vect{\Gamma}^{-1}_{\phi^c_l} \{- \dot{\vect{\Gamma}}_{\phi^c_l} \vect{e}^{\phi^c_l} + \vect{\Lambda}_{\Hat{\vect{x}}_{\phi^c_l}}^{\top}\dot{\Hat{\vect{x}}}_{\phi^c_l} + \vect{\Lambda}_{\vect{x}_{\phi^c_l}}^{\top} [\vect{f}_{\phi^c_l}+\\  & \vect{w}_{\phi^c_l}] \} - \vect{\epsilon}^{\phi^c_l\top} \vect{J}_{\phi^c_l} \vect{\Gamma}^{-1}_{\phi^c_l} \vect{\Lambda}_{\vect{x}_{\phi^c_l}}^{\top} \Theta_{\phi^c_l}  \vect{\Lambda}_{\vect{x}_{\phi^c_l}} \vect{\Gamma}^{-1}_{\phi^c_l} \vect{J}_{\phi^c_l}\vect{\epsilon}^{\phi^c_l},
        \end{split}
    \end{equation}
    where $\vect{J}_{\phi^c_l}$, $\vect{\Gamma}^{-1}_{\phi^c_l}$ and $\vect{\Lambda}_{\vect{x}_{\phi^c_l}}$ are defined as in \eqref{eq: transformed normalized error dynamics}, and $\Theta_{\phi^c_l}$ as in \eqref{eq: Lyapunov function expression}. 
    As in Step~A, under Assumption~\ref{Assumption on Task graph}, $\vect{\Lambda}_{\vect{x}_{\phi^c_l}}$ can be written as a block \mbox{upper-triangular} matrix with $\frac{\partial \hat{\rho}^{\psi_{l_i}}}{\partial x_{l_i}}$, $l_i \in \mathcal{V}_l$, on the main diagonal.
    % By Assumption \ref{Assumption on concavity and weelposeness of the robustness function}, whenever ${\rho}^{\psi_{l_i}} < \rho^{\text{opt}}_{\psi_{l_i }}$, $\frac{\partial {\rho}^{\psi_{l_i}}}{\partial x_{l_i}} \neq 0$ holds. 
    Then, since $\rho^{\max}_{\psi_{l_i }} <\rho^{\text{opt}}_{\psi_{l_i }}$ by design, and every task depends on the true state of the agents to which it is assigned, $\text{rank}(\vect{\Lambda}_{\vect{x}_{\phi^c_l}}) = v_l$ and $\vect{\Lambda}_{\vect{x}_{\phi^c_l}}^{\top} \Theta_{\phi^c_l} \vect{\Lambda}_{\vect{x}_{\phi^c_l}} \succ 0$. Since $J_{\psi_{l_i}}(e^{\psi_{l_i}}) >0$ and $\Gamma_{\psi_{l_i}}(t)>0$, $- \vect{\epsilon}^{\phi^c_l\top} \vect{J}_{\phi^c_l} \vect{\Gamma}^{-1}_{\phi^c_l} \vect{\Lambda}_{\vect{x}_{\phi^c_l}}^{\top} \Theta_{\phi^c_l}  \vect{\Lambda}_{\vect{x}_{\phi^c_l}} \vect{\Gamma}^{-1}_{\phi^c_l} \vect{J}_{\phi^c_l}\vect{\epsilon}^{\phi^c_l} <  - \alpha_{\rho} \alpha_J \vect{\epsilon}^{\phi^c_l \top}\vect{\epsilon}^{\phi^c_l}$ holds with $\alpha_{\rho} = \lambda_{\text{min}}(\vect{\Lambda}_{\vect{x}_{\phi^c_l}}^{\top} \Theta_{\phi^c_l} \vect{\Lambda}_{\vect{x}_{\phi^c_l}}) \in \mathbb{R}$ and $\alpha_J = \min_{l_i \in \mathcal{V}_l} \Bigl\{\min_{t \in \mathbb{R}_{\geq 0}, e^{\psi_{l_i}}\in (-1,0)} \Gamma_{\psi_{l_i}}^{-2}(t) J^2_{\psi_{l_i}}(e^{\psi_{l_i}}) \Bigr\}$.
    
    By adding and subtracting $\zeta \norm{\vect{\Gamma}^{-1}_{\phi^c_l} \vect{J}_{\phi^c_l}\vect{\epsilon}^{\phi^c_l}}^2$ for some $0 < \zeta < \alpha_\rho$, \eqref{eq: Lyapunov function expression in inner clusters} can be rewritten as:
    \begin{equation}
        \label{eq: Lyapunov function final upper bound, observation case}
        \dot{V} \leq - \kappa V + \vect{b}(t),
    \end{equation}
    where $\kappa = 2(\alpha_{\rho} - \zeta) \alpha_J$ and $\vect{b}(t) = \frac{1}{4\zeta} (\norm{\vect{\Lambda}_{\vect{x}_{\phi^c_l}}}\norm{\vect{f}_{\phi^c_l}} + \norm{\vect{\Lambda}_{\vect{x}_{\phi^c_l}}}\norm{\vect{w}_{\phi^c_l}} + \norm{\dot{\vect{\Gamma}}_{\phi^c_l}  \vect{e}^{\phi^c_l}}+ \norm{\vect{\Lambda}_{\Hat{\vect{x}}_{\phi^c_l}}} \norm{\dot{\Hat{\vect{x}}}_{\phi^c_l}})^2$. Since \eqref{eq: Lyapunov function final upper bound, observation case} resembles \eqref{eq: Lyapunov function final upper bound}, to prove the satisfaction of the cluster's task, it suffices to prove the existence of an upper bound $\Bar{\vect{b}}(t)$ on $\vect{b}(t)$ as in Step~A. Consider the set $\mathcal{X}_{l_i}(t):=\{\hat{\vect{x}}^\top_{\phi_{l_i}}\in \mathbb{R}^{N^{TC}_{i}}| -1 < e^{\psi_{l_i}} < 0 \}$ containing the states $\Bar{\vect{x}}_{\phi_{l_i}}$ and the state estimates $\hat{\vect{x}}^{l_i}_{\phi_{l_i}}$ satisfying task $\phi_{l_i}$. Following Lemma~\ref{lemma on prescribed performance function}, suppose $\Gamma_{\phi_{l_i}}(t)$ assumes its maximum value in $t = t_{\text{max}}$. Then, $\mathcal{X}_{l_i}(t) \subseteq \mathcal{X}_{l_i}(t_{\text{max}})$ for all $t$, and $\mathcal{X}_{l_i}(t_{\text{max}})$ collects all states $\Hat{\vect{x}}_{\phi_{l_i}}$ such that $e^{\psi_{l_i}} \in (-1,0)$ at all $t \in \mathbb{R}_{\geq 0}$. As discussed in Remark~\ref{Remark: weel-posenes of robustness}, since the robustness depends on state estimation, condition (iii) in Assumption~\ref{Assumption on concavity and weelposeness of the robustness function} cannot be guaranteed simply by adding $\psi^{\text{Bound}}_{l_i} := (\norm{\vect{x}_{\phi_{l_i}}}<\Bar{C}_{l_i})$ to $\phi_{l_i}$. However, since the out-neighboring clusters of $\mathcal{C}_l$, i.e., $\mathcal{C}_j$ with $(\mathcal{C}_l, \mathcal{C}_j) \in \mathcal{E}'$, complete their tasks as established in Step~A, augmenting the original $\psi_{l_i}$ and $\psi_{j_q}$ with $\psi^{\text{Bound}}_{l_i} := (\norm{\Bar{\vect{x}}_{\phi_{l_i}}}<\Bar{C}_{l_i})$ and $\psi^{\text{Bound}}_{j_q} := (\norm{\Bar{\vect{x}}_{\phi_{j_q}}}<\Bar{C}_{j_q})$, for all $l_i \in \mathcal{V}_l$ and $j_q \in \mathcal{V}_j$ with $(\mathcal{C}_l, \mathcal{C}_j) \in \mathcal{E}'$, together with the use of the observer \eqref{observer dynamics}, ensures that 
    $\psi_{l_i} \land \psi^{\text{Bound}}_{l_i}$ and $\psi_{j_q} \land \psi^{\text{Bound}}_{j_q}$ are \mbox{well-posed}. 
    Thus, if $\Bar{C}_{l_i},  \Bar{C}_{j_q}\in \mathbb{R}$ are selected big enough to preserve feasibility, $\mathcal{X}_{l_i}(t)$ remains bounded in time and $ \mathcal{X}_{l_i}(t_{\text{max}})$ is guaranteed to be a bounded open set. As a result, for the same reasons as in Step~A, $\norm{\vect{f}_{\phi^c_l}}$, $\norm{\vect{w}_{\phi^c_l}}$, $ \norm{\dot{\vect{\Gamma}}_{\phi^c_l}  \vect{e}^{\phi^c_l}}$, $\norm{\vect{\Lambda}_{\vect{x}_{\phi^c_l}}}$ and $\norm{\vect{\Lambda}_{\Hat{\vect{x}}_{\phi^c_l}}}$ are bounded on $\mathcal{X}_{l_i}(t_{\text{max}})$. $\dot{\Hat{\vect{x}}}_{\phi^c_l}$ is a vector containing the dynamics of all estimates involved in tasks $\phi_{l_i}$, i.e., $\dot{\hat{x}}_{r}^{l_i} = - \rho_{r}^{l_i}(t)^{-1} J(e_{r}^{l_i}) \epsilon_{r}^{^{l_i}}(t)$, for all $l_i \in \mathcal{V}_l$ and $r \in  \Neigh{{l_i}}{T} \setminus \Neigh{{l_i}}{C}$. From the observer guarantees, $|\tilde{x}_{r}^{l_i}| < \delta_{r}^{l_i}$ holds. Thus, $\epsilon_{r}^{l_i} < \infty$, $J(e_{r}^{l_i}) < \infty$, and each component $\dot{\hat{x}}_{r}^{l_i}$ of $\dot{\Hat{\vect{x}}}_{\phi^c_l}$ is bounded. Hence, $\norm{\dot{\Hat{\vect{x}}}_{\phi^c_l}}$ is bounded as well.  Given the previous bounds, an upper bound ${\Bar{\vect{b}}}(t)$ on $\vect{b}(t)$ is guaranteed to exist for all $\hat{\vect{x}}_{\phi_{l_i}} \in \mathcal{X}_{l_i}(t_{\text{max}})$ and all $l_i \in \mathcal{V}_l$. 
    By introducing $S(\vect{e}^{\phi^c_l}(t))$, task satisfaction can be proven as in Step~A. For this reason, the rest of the proof is omitted for brevity.
    
    % In Step B, we considered clusters $\mathcal{C}_l$ whose \mbox{out-neighbors} are leaf. A similar procedure to the one proposed here can be applied to verify task satisfaction for the remaining ones. Owing to the acyclicity of $\mathcal{G}'$ and assuming task satisfaction for the already treated clusters, it is then possible to establish the same property for all \mbox{higher-level} nodes, up to and including the roots. {\hfill$\Box$\par}
   In Step B, we analyzed clusters $\mathcal{C}_l$ whose out-neighbors are leaves. Owing to the acyclic structure of $\mathcal{G}'$, and assuming task satisfaction for the already treated clusters, the same reasoning can be applied iteratively to ensure task satisfaction all the way up to the root clusters.{\hfill$\Box$\par}
\end{pf}

\begin{remark}
    If $g_{l_i}$ in \eqref{eq: agent's dynamic} is a square symmetric matrix with known sign, then $u_{l_i} = - s_{l_i} \sum_{l_j\in \mathcal{V}_l} \frac{\partial {\rho}^{\psi_{l_j}}(\hat{\vect{x}}_{\phi_{l_j}})}{\partial x_{{l_i}}} \Gamma^{-1}_{\psi_{l_j}}  J_{\psi_{l_j}} \epsilon^{\psi_{l_j}}$ ensures Problem \ref{Problem: first problem formulation} to be solved with $s_{l_i} = 1$ if $g_i$ is positive definite and $s_{l_i} = -1$ otherwise.
\end{remark}

\section{Simulation}\label{Section: Simulation}
\begin{figure}[t!]
    \centering
    \includegraphics[width =0.9\linewidth]{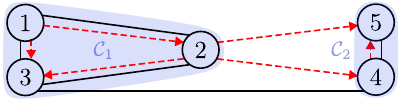}
    \caption{Graphs $\mathcal{G}_C$ and $\mathcal{G}_T$, respectively in solid and dashed lines. The shadow areas represent the induced clusters, i.e., $\mathcal{C}_1$ and $\mathcal{C}_2$.}
    \label{fig:Graph Gt used for consensus}
\end{figure}

Consider the MAS in Fig.~\ref{fig:Graph Gt used for consensus}, composed of $N=5$ agents communicating according to the connected graph $\mathcal{G}_C$. Suppose each agent behaves as $\dot x_i =  f(x_i) + g(x_i) u_i + w_i(t)$, where $x_i =[x_{i,1}, x_{i,2}]^\top$ and $u_i \in \mathbb{R}^2$ are the state and input vectors, $f(x_i) = [x_{i,1} +2x_{i,2} + 0.6 \arctan(x_{i,1}) + 0.3\tanh(0.8 x_{i,2}), 3x_{i,1} +4x_{i,2} + 0.5 \sin(0.7 x_{i,1} + 0.2 x_{i,2}) + 0.4 \arctan(0.5 x_{i,2})]^\top$, $g(x_i) = \begin{bmatrix}\cos(0.5x_{i,2}) & - \sin(0.5x_{i,2}) \\  \sin(0.5x_{i,2}) & \cos(0.5x_{i,2})\end{bmatrix}$, and $w_i(t)= [w_{i,1}, w_{i,2}]^\top$ is a random disturbance with $w_{i,1}, w_{i,2} \in [-6,6]$. 

Assume the agents are subject to the following STL tasks: $\phi_1 = G_{[1,2]}((\norm{x_1 - [0,2]^\top} \leq 7)\land (\norm{x_1 - x_2}^2 \leq 26.75) \land (\norm{x_1 - x_3}^2 \leq 70.05))$,  $\phi_2 = G_{[1,2]}((\norm{x_2 - x_3}^2 \leq 26.75) \land (\norm{x_2 - x_4}^2 \leq 70.05) \land  (\norm{x_2 - x_5}^2 \leq 70.05))$,  $\phi_3 = G_{[1,2]}(\norm{x_3 - [-1.175, -1.618]^\top}^2 \leq 7)$, $\phi_4 = G_{[1,2]}(\norm{x_4 - x_5}^2 \leq 26.75)$, $\phi_5 = G_{[1,2]}(\norm{x_5 - [1.9, 0.618]^\top}^2 \leq 7)$. Intuitively, all agents are required to get sufficiently close to each other, while allowing agents $1$, $3$ and $5$ to reach their goal positions.
From the task and communication graph, respectively $\mathcal{G}_T$ and $\mathcal{G}_C$ in Fig.~\ref{fig:Graph Gt used for consensus}, the MAS is split into $N' = 2$ clusters, i.e., $\mathcal{C}_1$ and $\mathcal{C}_2$ with $\mathcal{V}_1 = \{1,2,3\}$ and $\mathcal{V}_2 = \{4,5\}$. 
Since agents $4$ and $5$ are not in communication with $2$, to allow the satisfaction of $\phi_2$ we leverage the observer introduced in Section \ref{Section: Distributed observer} with $k = 3$. To guarantee the satisfaction of the tasks by means of the controller \eqref{eq: agent's input}, the function $\gamma_{2}$ is adjusted, following Example \ref{example: funnel modification}, to account for the \mbox{worst-case} estimation errors. Given $\mathcal{G}_C$ and $\mathcal{G}_T$ in Fig.~\ref{fig:Graph Gt used for consensus}, by proper system initialization and observer specifications, Assumptions~\ref{Assumption on neighbors}-\ref{Assumption on initialization} are satisfied. Thus, Theorem~\ref{Theorem on task satisfacion} holds and \eqref{eq: agent's input} guarantees tasks satisfaction.

Fig.~\ref{fig: Simulation results}(a) shows the MAS trajectory, along with the time evolution of the state estimates required for task $\phi_2$ and the corresponding uncertainty regions due to estimation errors. Due to space limitations, Figs.~\ref{fig: Simulation results}(b)-(d) present only the results for agent~$2$. As illustrated in Figs.~\ref{fig: Simulation results}(c)-(d), enforcing  $\min_{j \in \{1,2,3\}} \hat{\rho}^{\psi_{2,j}} > 0$ for all $t \in [1,2]$ by imposing \eqref{eq: inequality on robustness to guarantee task satisfaction, observer not in design stage of gamma}, with $\Gamma_{\psi_2}(t) =\min_{j \in \{1,2,3\}} \Gamma_{\psi_{2,j}}$, is sufficient to guarantee satisfaction of $\phi_2$. Indeed, Fig.~\ref{fig: Simulation results}(d) shows that $\rho^{\psi_{2,j}} > 0$ for all $j \in \{1,2,3\}$ and $t\in [1,2]$.

\begin{figure}[t!]
    \centering
    \includegraphics[width =\linewidth]{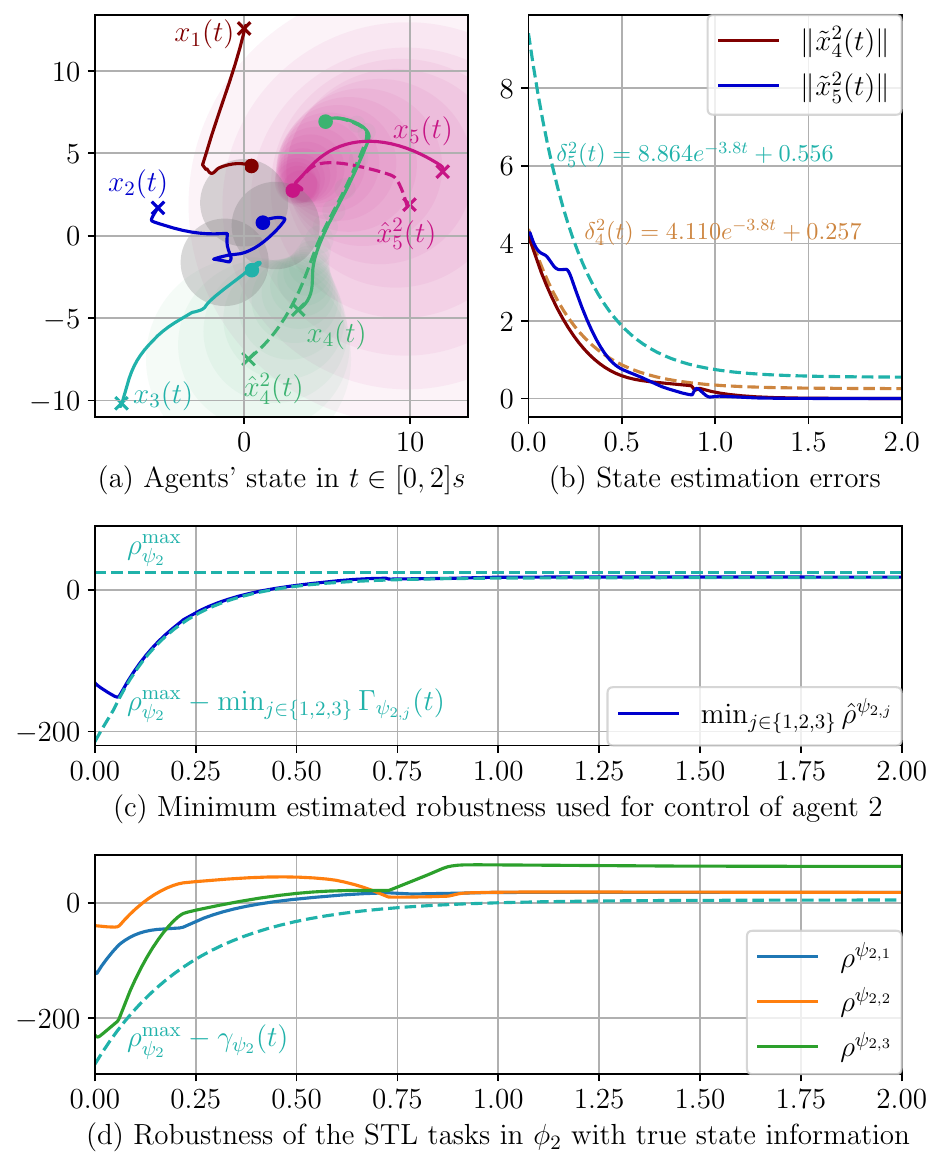}
    \caption{(a) State and estimation trajectories; initial states are represented by crosses, terminal state by dots. The pink and green areas represent the evolution of the state uncertainties with time; the gray circles are the target areas of the individual tasks. (b) Evolution of the norm of the state estimation errors performed by agent $2$; the dashed lines are the prescribed performance functions of the $k$-hop PPSO. (c) Minimum estimated robustness of $\phi_2$. (d) Robustness of tasks in $\phi_2$ with true state information.}
    \label{fig: Simulation results}
\end{figure}

\section{Conclusion and Future Work}\label{Section: Conclusion}
% This paper proposes a decentralized controller for \mbox{large-scale} heterogeneous \mbox{multi-agent} systems subject to bounded external disturbances, where agents must satisfy Signal Temporal Logic (STL) specifications requiring cooperation among non-communicating agents. To address the lack of direct communication, we employ a decentralized $k$-hop Prescribed Performance State Observer ($k$-hop PPSO) to provide each agent with state estimates of those it cannot communicate with directly. By leveraging the performance bounds on the state estimation errors imposed by the $k$-hop PPSO, we first modify the tasks to account for these errors, and then design a decentralized \mbox{closed-form} \mbox{continuous-time} feedback controller that ensures task satisfaction even under \mbox{worst-case} state estimation errors. A simulation result is provided to validate the proposed framework.
\vspace{-0.2cm}
We proposed a decentralized \mbox{observer-based} controller for MAS subject to STL specifications under communication/task graphs mismatch. Given its structure and decentralized nature, the proposed solution is robust to bounded external disturbances and suitable for \mbox{large-scale} heterogeneous systems. Future works will address the \mbox{observer-controller} \mbox{co-design} to relax Assumption~\ref{Assumption on the non-existance of a task not in communication inside a cluster}. 
\bibliography{ifacconf}             % bib file to produce the bibliography
                                                     % with bibtex (preferred)
                                                   
%\begin{thebibliography}{xx}  % you can also add the bibliography by hand

%\bibitem[Able(1956)]{Abl:56}
%B.C. Able.
%\newblock Nucleic acid content of microscope.
%\newblock \emph{Nature}, 135:\penalty0 7--9, 1956.

%\bibitem[Able et~al.(1954)Able, Tagg, and Rush]{AbTaRu:54}
%B.C. Able, R.A. Tagg, and M.~Rush.
%\newblock Enzyme-catalyzed cellular transanimations.
%\newblock In A.F. Round, editor, \emph{Advances in Enzymology}, volume~2, pages
%  125--247. Academic Press, New York, 3rd edition, 1954.

%\bibitem[Keohane(1958)]{Keo:58}
%\newblock \emph{Power and Interdependence: World Politics in Transitions}.
%\newblock Little, Brown \& Co., Boston, 1958.

%\bibitem[Powers(1985)]{Pow:85}
%\newblock Is there a way out?
%\newblock \emph{Harpers}, pages 35--47, June 1985.

%\bibitem[Soukhanov(1992)]{Heritage:92}
%\newblock \emph{{The American Heritage. Dictionary of the American Language}}.
%\newblock Houghton Mifflin Company, 1992.

%\end{thebibliography}

\end{document}